\documentclass[draftcls,onecolumn,a4paper]{IEEEtran}

\usepackage[english]{babel}
\usepackage{amssymb,amsmath,amscd,amsfonts,amsthm,bbm,bm}
\usepackage{graphicx}
\usepackage{psfrag}
\usepackage{paralist}
\usepackage{color}
\usepackage{float}
\usepackage{relsize}
\usepackage{euscript}
\usepackage{setspace}
\usepackage{enumerate}
\usepackage{flushend}
\usepackage{tikz}
\usetikzlibrary{plotmarks}
\usepackage{pgfplots}
\usetikzlibrary{arrows,shapes,decorations,backgrounds}
\usepackage{tkz-berge}
\usepackage{calc}

\newcommand{\revju}[1]{#1}
\newcommand{\revjul}[1]{#1}

\DeclareMathOperator{\tr}{Tr}
\DeclareMathOperator{\support}{supp}
\DeclareMathOperator{\rank}{rank}
\DeclareMathOperator{\diag}{diag}
\newcommand{\bs}{\boldsymbol}

\newcommand{\toaslong}{\xrightarrow[T\to\infty]{\text{a.s.}}}
\newcommand{\toasshort}{\xrightarrow{\text{a.s.}}}
\newcommand{\toprobalong}{\xrightarrow[T\to\infty]{{\cal P}}}
\newcommand{\toprobashort}{\xrightarrow{{\mathcal P}}}
\newcommand{\tolawlong}{\xrightarrow[T\to\infty]{{\cal L}}}
\newcommand{\tolawshort}{\xrightarrow{{\mathcal L}}}
\newcommand{\tolong}{\xrightarrow[T\to\infty]{}}

\newcommand{\N}{\mathbb N}
\newcommand{\R}{\mathbb R}
\newcommand{\Z}{\mathbb Z}
\newcommand{\C}{\mathbb{C}}
\newcommand{\E}{\mathbb{E}}
\newcommand{\1}{\mathbbm 1}

\newcommand{\bl}{\{} 
\newcommand{\br}{\}} 

\def\bx{{\bf x}}\def\bx{{\bf x}}

\def\bA{{\bf A}}
\def\bY{{\bf Y}}
\def\bV{{\bf V}}
\def\bQ{{\bf Q}}
\def\ba{{\bf a}}
\def\bc{{\bf c}}
\def\bd{{\bf d}}
\def\bm{{\bf m}}
\def\bg{{\bf g}}
\def\bp{{\bf p}}
\def\bv{{\bf v}}
\def\bx{{\bf x}}
\def\by{{\bf y}}
\def\brho{{\boldsymbol \rho}}
\def\bnu{{\boldsymbol \nu}}
\def\bmu{{\boldsymbol \mu}}
\def\bDelta{{\boldsymbol \Delta}}
\def\bdelta{{\boldsymbol \delta}}

\def\bOmega{{\bf \Omega}}
\def\bbT{{\bf \mathbb{T}}}

\newtheorem{lemma}{Lemma}

\newtheorem{theorem}{Theorem}
\newtheorem{corollary}{Corollary}
\newtheorem{remark}{Remark}
\newtheorem{assumption}{Assumption}
\newtheorem{proposition}{Proposition}

\usepackage{soul}

\begin{document}

\title{Statistical Inference in Large Antenna Arrays \\ under Unknown Noise Pattern}

\author{Julia Vinogradova, Romain Couillet, and Walid Hachem
\thanks{Copyright (c) 2013 IEEE. The first and third authors are with CNRS LTCI; Telecom ParisTech, 
France (e-mail: julia.vinogradova@telecom-paristech.fr; walid.hachem@telecom-paristech.fr). The second author is with Supelec, France (e-mail: romain.couillet@supelec.fr). 
This work is supported in part by the French Ile-de-France region, DIM LSC fund, Digiteo project DESIR, and in part by the ANR-12-MONU-OOO3 DIONISOS.
}}

\date{\today} 
  
\maketitle

\begin{abstract}
	In this article, a general information-plus-noise transmission model is assumed, the receiver end of which is composed of a large number of sensors and is unaware of the noise \revju{correlation} pattern. For this model,\revju{under an isotropy assumption between signal and noise {left- and right-}eigenspaces,} a set of results is provided for the receiver to perform statistical eigen-inference on the information part. In particular, we introduce new methods for the detection, counting, and the power and subspace estimation of multiple sources composing the information part of the transmission. The theoretical performance of some of these techniques is also discussed. An exemplary application of these methods to array processing \revju{with unknown time correlated noise} is then studied in greater detail, leading to a novel MUSIC-like algorithm.
\end{abstract}
\vspace{.3mm}

\section{Introduction}\label{intro} 
\def\herm{{\sf H}}
\def\CC{{\mathbb C}}
\def\EE{{\mathrm E}}
\def\asto{\overset{\rm a.s.}{\longrightarrow}}

\subsection{Motivation}

Consider the information-plus-noise transmission model with multivariate output $y_t\in\CC^N$ at time $t$ 
\begin{equation}
\label{eq:general_model}
y_t = H x_t + v_t
\end{equation}
where $x_t\in\CC^K$ is the vector of transmitted symbols at time $t$, $H\in\CC^{N\times K}$ is the linear communication medium, and $v_t\in\CC^N$ the noise experienced by the receiver at time $t$.

Array processing consists in a set of tools to perform statistical inference on the information part \revju{$Hx_t$} composing $y_t$. The first tool is the mere detection of this information (called then a signal source), that is the question whether $K>0$. Once source signals are detected, the next operation consists in the evaluation of their number, {\it i.e.} estimating $K$. When the existence of these sources is guaranteed, several of their parameters can then be retrieved. One of these parameters is the transmission power of the source or, alternatively, the distance from the source to the receiver. Denoting $H=[h_1,\ldots,h_K]$, it is also of interest to retrieve information from the individual $h_k$ vectors. In wireless communications, these represent channel \revju{vectors} which the receiver may want to identify in order to decode the entries of $x_t$. In array processing, they stand for steering vectors parameterized by the angle-of-arrival of the source signals. 

In order to perform these tasks, one assumes the observation of $T$ (non-necessarily independent) samples $y_1,\ldots,y_T$ of the process $y_t$. Denoting $Y_T=T^{-1/2}[y_1,\ldots,y_T]$, the first mentioned estimators are often based on the eigenvalues of $Y_TY_T^{\sf H}$. When it comes to vector identification, the interest is rather on the eigenvectors of $Y_TY_T^{\sf H}$. The standard eigen-inference approaches in the literature often rely on two strong assumptions: (i) $T$ is large compared to $N$ and (ii) the statistics of $v_t$ are partially or perfectly known due to independent (information-free) observations of the process $v_t$. This article revisits these methods by proposing alternative algorithms to perform eigen-inference for the model \eqref{eq:general_model} accounting for the aforementioned limitations (i) and (ii).

\subsection{Literature review}

Assuming $T\to\infty$, $N$ fixed, and $v_t$ white Gaussian with known variance, the energy detection procedure \cite{URK67} allows for the detection of signal sources by evaluating the total received power which is compared to a threshold that ensures \revju{an admissible} false alarm rate. If the signal structure is known, the parameters composing $H$ can be recovered from the eigenvalues and eigenvectors of $\EE[y_ty_t^\herm]$, which can be estimated through the sample covariance matrix $Y_TY_T^\herm$, $Y_T=T^{-1/2}[y_1,\ldots,y_T]\in\CC^{N\times T}$. To estimate the number of sources $K$, the Akaike information criterion (AIC) \cite{Akai'74} and the minimum description length (MDL) \cite{Riss'78,WaxKai'85} were historically proposed, which rely on functions of the eigenvalues of $Y_TY_T^\herm$. The MDL is $T$-consistent while the AIK tends to overestimate the number of sources as $T\to\infty$. In terms of power estimation, since $Y_TY_T^\herm\asto \EE[y_ty_t^\herm]$, a $T$-consistent estimate of the powers is easily obtained by mapping the eigenvalues of $Y_TY_T^\herm$ to those of $\EE[y_ty_t^\herm]$. When the vectors $h_k=h(\theta_k)$ are steering vectors and that one aims at retrieving $\theta_k$ for $k=1,\ldots,K$, the multiple signal classification (MUSIC) algorithm \cite{Schm'86} allows for a $T$-consistent estimation of the angles $\theta_1,\ldots,\theta_K$ by determining the local maxima of the quadratic forms $\gamma(\theta)=h(\theta)^\herm \Pi h(\theta)$ where $\Pi$ is a projector on the eigenspace of $\EE[y_ty_t^\herm]$ corresponding to its $K$ largest eigenvalues (assuming $\Vert h(\theta)\Vert$ constant with $\theta$).

Due to the increase of the antenna array sizes and the need for faster detection and estimation dynamics, modern antenna array technologies have to deal with the scenario where the condition $T\gg N$ is no longer met. Under this condition, since $Y_TY_T^\herm$ becomes a poor estimator for $\EE[y_ty_t^\herm]$, most of the above techniques collapse. New methods, based on the field of large dimensional random matrix theory, have therefore emerged, which assume that both $N$ and $T$ are large and that the ratio $N/T$ is non-trivial. The AIC and MDL algorithms are in particular improved in \cite{NadaEde'08} using better estimators for functionals of the eigenvalues of $\EE[y_ty_t^\herm]$. \revjul{Another non-parametric approach based on hypothesis testing which provides a refined asymptotic detection limit was proposed in \cite{KritNad'09}. A parametric-based algorithm with estimation of an unknown noise variance was developped in \cite{KritNad'08}.} In terms of power estimation, $N,T$-consistent techniques were proposed in \cite{COU10b}. The MUSIC algorithm was improved on the same grounds in \cite{mestre2008improved} into the so-called G-MUSIC estimator.

A second difficulty faced by antenna array technologies is that the interfering environment may be far from white Gaussian. The $v_1,\ldots,v_T$ may not be independent or the spatial correlation of $v_t$ may not be white. When the noise is not white, the energy detection procedure is not valid as no false alarm threshold can be set. When the noise is close-to-white Gaussian with unknown variance, the generalized likelihood-ratio test (GLRT) \cite{BessKraSch'06} copes with the indetermination of the variance. Similar schemes are analyzed in the large $N,T$ regime in \cite{card-08,penna-cl09, BianDebMai'11,nad-icc11}. If the noise is not white, it is difficult to derive any test for detection. The power and direction estimation techniques equally suffer from this indetermination, because too little is a priori known of the eigenstructure of $V_TV_T^\herm$ with $V_T=T^{-1/2}[v_1,\ldots,v_T]$. To circumvent this issue, one generally assumes the existence of a sequence of $T'$ pure-noise test samples which are used to ``whiten'' the observations. For $T'$ large compared to $N$, after whitening, the noise becomes white Gaussian with unit variance, leading back to traditional schemes. For $N,T'$ simultaneously large, the whitening procedure gives rise to a noise matrix of the $F$-matrix type \cite{NadaSil'10,nad-joh-ssp11}.

However, the requirement to possess observations purely composed of noise may be impractical in real systems. As such, in this article, we address the problems of detection, counting, and parameter estimation of multiple sources without resorting to a pre-whitening of the received data matrix $Y_T$. Since the problem may not be well-posed in its generality, we assume a set of reasonable conditions:
\begin{itemize}
	\item $N,T\to\infty$, $N/T\to \bc>0$, $K$ constant. This allows for $Y_TY_T^\herm$ to be seen as a small rank perturbation of $V_TV_T^\herm$.
	\item $V_T=W_TR_T^{1/2}$ ({\it i.e.} white in space, correlated in time), where $W_T\in\CC^{N\times T}$ is standard complex Gaussian and $R_T$ is a deterministic {\it unknown} Hermitian nonnegative, or $V_T=R_T^{1/2}W_T$ ({\it i.e.} white in time, correlated in space).
	\item As $N/T\to \bc$, the eigenvalues of $V_TV_T^\herm$ tend to cluster in a compact interval. This assumption is satisfied by most noise models used in practice, {\it e.g.} auto-regressive moving average (ARMA) noise processes (see Section~\ref{example-doa}).
	\item \revju{If $V_T$ is correlated in time, }the source signals in $x_t$ are random, independent, and identically distributed (i.i.d.)\revju{ while, if $V_T$ is correlated in space, $H_T$ is not correlated in space}.
\end{itemize}
Note in particular that the scenario where signal and noise are correlated both in space or both in time cannot be addressed. This is linked to the fact that, whenever the (left- or right-) eigenspaces of $R_T$ affect the signal parameters to be estimated, one needs information on these eigenspaces which is in general too demanding from a single observation of $Y_T$ (unless more structural information on $R_T$ is available which we do not assume). Instead, we require here that, at least asymptotically, the parameters to estimate only depend on the eigenvalues of $R_T$ which can be inferred from $Y_T$. Note also that $V_T$ cannot be correlated in both time and space, which would lead to a so far too difficult problem to address with the existing random matrix tools. There exist several practical scenarios in which those hypotheses are valid \textit{e.g.} in civil radars with interfering (non-radar) signals. As civil radar beams (typically in open spaces, \textit{e.g.} in urban environment) are usually very directive, signal multi-path is not expected (or expected to be weak). Interference arising from surrounding electromagnetic fields may however typically contain multi-path. In a dense scattering environment, the induced noise would therefore be loosely directive in space but correlated in time, accordingly with our model. Line-of-sight communications subject to spatially white multi-path interference can be cited as well. In fact, localization in wireless communications is only effective if the users are localized in line-of-sight of the exploring base station, therefore implying weak multi-path reflections of the signals of interest. Being subject to multi-user interference from their own or adjacent cells, the signals received at the base station therefore contain multi-path interference, again in line with our assumption.

Under these assumptions, we show that a maximum of $K$ isolated eigenvalues of $Y_TY_T^\herm$ can be found for all large $N,T$ beyond the right edge of the limiting eigenvalue distribution support of $V_TV_T^\herm$. This phenomenon is at the origin of the detection and estimation procedures developed in this paper. Precisely, we show that the isolated eigenvalues of $Y_TY_T^\herm$ can be uniquely mapped to individual signal sources. The presence of these eigenvalues will be used to detect signal sources as well as to estimate their number $K$ while their values will be exploited to estimate the source powers. The associated eigenvectors will then be used to retrieve information on the vectors $h_k$.  

The remainder of the article is structured as follows. In Section~\ref{sec:assumptions}, we introduce the system model and recall important results from the random matrix literature. In Section~\ref{1st-ord}, we introduce the source detector and parameter estimators for the generic model \eqref{eq:general_model} and for a specific array processing scenario with an ARMA noise process. In Section~\ref{2nd-ord}, we study the second order statistics of some of these estimators. Simulations are then provided in Section~\ref{simulation}. The article is concluded by Section~\ref{sec:conclusion}. Some technical lemmas are proved in the appendix. 

{\it Notations:} The superscript $(\cdot)^{\sf H}$ is the Hermitian transpose of a matrix and $\left\|\cdot\right\|$ denotes the spectral norm. The symbols 
$\overset{\text{a.s.}}{\longrightarrow}$,
$\overset{{\mathcal P}}{\longrightarrow}$, and
$\overset{{\mathcal L}}{\longrightarrow}$
stand respectively for the almost sure convergence, the convergence in probability, and the convergence in law, while ``w.p. 1'' means ``with probability one''. We denote by ${\cal N}(a,\sigma^2)$ the real Gaussian distribution with mean $a$ and variance $\sigma^2$ and by ${\cal CN}(a, \sigma^2)$ the complex circular Gaussian distribution with mean $a$ and variance $\sigma^2$. We denote by $\bdelta_{k\ell}$ the Kronecker delta function ($=1$ if $k=\ell$ and $0$ otherwise) and by $\delta_x$ the Dirac measure at $x$. \revju{Finally, boldface characters denote limiting values.}

\section{Assumptions and known results}
\label{sec:assumptions}

Consider a sequence of integers $N = N(T)$, $T = 1, 2, \ldots$ and matrices $Y_T = A_T + W_T R_T^{1/2}\in\CC^{N\times T}$ where $A_T$ stands for the signal matrix and $V_T = W_T R_T^{1/2}$ for the noise matrix. \revju{Remark that, up to studying $Y_T^\herm$ instead of $Y_T$, the noise correlation can be either in time or in space.} We assume the following asymptotic regime: 
\begin{assumption}
\label{ass:c}
As $T\to\infty$, $c_T \triangleq N/T\to {\bc} > 0$.
\end{assumption}

\subsection{Hypotheses on the noise matrix}

We first characterize the assumptions on $V_T\triangleq W_T R_T^{1/2}$.

\begin{assumption}
\label{ass:gauss} 
$W_T = T^{-1/2} [ w_{n,t} ]_{n,t=1}^{N,T}$, with $(w_{n,t})_{n,t\geq 1}$ an infinite array of independent ${\cal CN}(0,1)$ variables. 
\end{assumption}

\begin{assumption}
\label{ass:R} 
$R_T\in\CC^{T\times T}$ is Hermitian nonnegative with eigenvalues $\sigma_{1,T}^2,\ldots,\sigma_{T,T}^2$ satisfying: 

\begin{enumerate}
\item \label{R:msl}
With $\nu_T = T^{-1} \sum_{t=1}^T \bdelta_{\sigma_{t,T}^2}$, $\nu_T\tolawshort \bnu$, a probability measure with support $\support(\bnu)=[a_\bnu, b_\bnu] \subset \R_+ \triangleq [0,\infty)$. Moreover, $\bnu(\bl 0 \br) = 0$. 

\item\label{R:noeig}
The distances from the $\sigma_{t,T}^2$ to $\support(\bnu)$ satisfy: 
\begin{equation*}
\max_{t\in\{1,\ldots,T\}} \bd\left(\sigma_{t,T}^2, \support(\bnu)\right)
\xrightarrow[T\to\infty]{} 0 . 
\end{equation*}
\end{enumerate}
\end{assumption}

Let $\lambda_{1,T}\geq \ldots \geq \lambda_{N,T}$ be the eigenvalues of $V_T V_T^\herm = W_T R_T W_T^{\sf H}$ and let $\tau_T = N^{-1} \sum_{i=1}^N \bdelta_{\lambda_{i,T}}$ be its spectral measure. The asymptotic behavior of $\tau_T$ is of prime importance in this paper. We recall some well known results describing this behavior; see \cite{MarcPas'67, SilvBai'95} for Items \ref{st-lsm})--\ref{cvg-unif}), \cite{SilvCho'95} for Item \ref{form-mu}), and \cite{BaiSil'98} for Item \ref{noeig-sil}): 
\begin{theorem} 
\label{th:lsm}
Under Assumptions \ref{ass:c}--\ref{ass:R}, the following hold true: 
\begin{enumerate}

\item
\label{st-lsm} 
For any $z \in \C_+ \triangleq \{ z \in \C,\, \Im z > 0\}$, the equation
\begin{equation} 
\label{m=f(m)} 
\bm = \left( -z + \int \frac{t}{1 + {\bc} \bm t} \bnu(dt) \right)^{-1} 
\end{equation}
has a unique solution $\bm \in \C_+$. The function $\bm(z) = \bm$ so defined on $\C_+$ is the Stieltjes transform (ST)\footnote{We recall that the ST $m_{\bmu}$ of a probability measure $\bmu$ with support in $\R$ is defined by $m_\bmu(z) = \int (t-z)^{-1} \bmu(dt)$. It is analytic on $\C - \support(\bmu)$ and completely characterizes the measure $\bmu$.} of a probability measure $\bmu$.

\item\label{cvg-Q} 
For every bounded and continuous real function $f$,
\begin{equation*}
\int f(t) \tau_T(dt) \ \toaslong \ \int f(t) \bmu(dt) 
\end{equation*}
and therefore $\bmu$, defined by \eqref{m=f(m)}, is the limiting spectral measure of $V_TV_T^\herm$.

\item\label{st-coresolv}
The function
\begin{equation*}
\tilde\bm(z) = \int \frac{-1}{z(1+\bc\bm(z) t)} \bnu(dt) 
\end{equation*}
is defined on $\C_+$ and is the ST of the probability measure $\tilde\bmu = \bc \bmu + (1-\bc) \delta_0$, limiting spectral measure of $V_T^\herm V_T$. As such, $\tilde\bm(z) = \bc\bm(z) - (1-\bc) /z$. 

\item\label{form-mu}
 $\bmu$ is of the form $\bmu(dt) = \max(0, 1 - {\bc}^{-1}) \delta_0 + f(t) dt$ where $f(t)$ is a continuous density on $(0,\infty)$. The support of $f(t)dt$ is a compact interval $[a,b] \subset \R_+$, and 
$f(t) > 0$ on $(a,b)$. 

\item\label{noeig-sil} 
For any interval $[x_1, x_2] \subset (0,a) \cup (b,\infty)$, 
\begin{equation*}
\sharp \{ i \, : \, \lambda_{i,T} \in [x_1, x_2] \} 
= 0 \ \text{w.p.} \ 1 \ \text{for \revju{all} large} \ T. 
\end{equation*}
\item\label{cvg-unif} 
The function $\underline m_T(x) = N^{-1} \sum_{n=1}^N (\lambda_{n,T} - x)^{-1}$ converges w.p. 1 to $\bm(x)$, and uniformly so on the compact subsets 
of $(b,\infty)$. 
% $\underline m_T(x) \toaslong \bm(x)$ uniformly on $[b', \infty)$. 

\end{enumerate} 
\end{theorem} 
A procedure for determining the interval $[a,b]$ from the knowledge of $\bc$ and $\bnu$ is provided in \cite{SilvCho'95}. \revju{In order to quantify the position of the rightmost eigenvalues of $V_TV_T^\herm$ (\emph{i.e.} noise only hypothesis),} we are interested here in the determination of the upper bound $b$, to which $\lambda_{1,T}$ converges. This can be done with the help of the following proposition. Observe that $\bm(z)$ can be extended to $\C - (\{0\} \cup [a,b])$ and that $\bm(x) = \int (t-x)^{-1} \bmu(dt)$, its restriction to $\R$, is negative and increases to zero on $(b, \infty)$. Recall that $\support(\bnu) = [a_\bnu, b_\bnu] \subset \R_+$.
\begin{proposition}[see \cite{SilvCho'95}]
\label{prop:edge}
The point $b$ defined in Theorem~\ref{th:lsm}-\ref{form-mu}) coincides with the infimum of the function
\begin{equation*}
\bx(m) = - \frac{1}{m} + \int \frac{t}{1+\bc m t} \, \bnu(dt)
\end{equation*}
on the interval $(-(\bc b_\bnu)^{-1}, 0)$.
On this interval, there is a unique $m_b$ ($m_b<0$) such that $\bx(m) \to b$
as $m \downarrow m_b$. The restriction of $\bx(m)$ to $(m_b, 0)$ coincides with the inverse with respect to composition
of the restriction of $\bm(x)$ to $(b,\infty)$.
\end{proposition}

In order to easily characterize the value of $b$, it will be convenient 
to make an assumption on the measure $\bnu$ which will not be restrictive 
in practice: 
% We note right now that a similar assumption was made by El Karoui 
% \cite{elkar'07} to characterize the fluctuations of the largest eigenvalues 
% of $W_T R_T W_T^{\sf H}$.  We will get back to this point in 
% Section \ref{2nd-ord}. 

\begin{assumption}
\label{nu:edge}
If $\bnu(\{b_\bnu\})=0$, then there exists $\varepsilon>0$ and a function $f_\bnu(t) \geq C(b_\bnu - t)$ on $[b_\bnu - \varepsilon, b_\bnu ]$ with $C > 0$ such that for any Borel set $A$ of $[a_\bnu, b_\bnu]$, 
\begin{equation*}
\bnu(A \cap [b_\bnu-\varepsilon, b_\bnu]) =  \int_{A \cap [b_\bnu-\varepsilon, b_\bnu]} f_\bnu(t) \, dt . 
\end{equation*} 
\end{assumption} 
%\begin{assumption}
%\label{nu:edge}
%There exist $\varepsilon > 0$, $\alpha \geq 0$ and a function 
%$f_\nu(t) \geq 0$ on $[b_\nu-\varepsilon, b_\nu]$ such that for any Borel 
%set $A$ of $[a_\nu, b_\nu]$, 
%\begin{equation*}
%\nu(A \cap [b_\nu-\varepsilon, b_\nu]) =  \alpha \delta_{b_\nu}(A) + 
%\int_{A \cap [b_\nu-\varepsilon, b_\nu]} f_\nu(t) \, dt . 
%\end{equation*} 
%If $\alpha = 0$, then there exists a constant $C > 0$ such that 
%$f_\nu(t) \geq C(b_\nu-t)$ on $[b_\nu-\varepsilon, b_\nu]$. 
%\end{assumption} 

\revju{The assumption states that $\bnu$ either has a mass or a sufficiently sharp density edge at $b_\bnu$. This assumption will be important in Section~\ref{2nd-ord} to determine the behavior of the proposed estimators close to the signal detectability limit. It presently }leads to the following corollary to Proposition \ref{prop:edge}, proven in Appendix \ref{prf:coredge}: 
\begin{corollary} 
\label{cor:edge} 
Under Assumption \ref{nu:edge},
\begin{equation*}
b = - \frac{1}{m_b} + \int \frac{t}{1+ \bc m_b t} \bnu(dt) 
\end{equation*}
where $m_b$ is the unique solution in $( -(\bc b_\bnu)^{-1}, 0)$ to the 
equation in $m$ 
\begin{equation} 
\label{eq:m(b+)} 
\int \left(  \frac{m t}{1 + \bc m t} \right)^2 \bnu(dt) = 
\frac{1}{\bc} . 
\end{equation} 
\end{corollary} 

\subsection{Hypotheses on the signal matrix}

We now turn to the hypotheses on the signal matrix $A_T$: 

\begin{assumption}
\label{ass:A}
Let $K\geq 0$ be a fixed integer. The matrix $A_T\in\CC^{N\times T}$ is random, independent of $W_T$, with rank $\rank(A_T) = K$ w.p. 1 for all large $T$. Besides, $\sup_T \| A_T \| < \infty$ w.p. 1. 
\end{assumption}
%\emph{Je suis oblig\'e de dire que le rang $=K$ avec proba 1 pour $T$ grand
%pour englober le cas QPSK}. 
\par In the remainder of the paper, when $K \leq \min(N,T)$, the notation 
$A_T = U_T B_T^{\sf H}$ refers to any factorization of $A_T$ where $U_T\in\CC^{N\times K}$ satisfies $U_T^{\sf H} U_T = I_K$. By Assumption~\ref{ass:A}, the rank of $B_T\in\CC^{T\times K}$ is equal to $K$, w.p. 1. \revju{We are now ready to make the fundamental assumption of the article:}
\begin{assumption} 
\label{ass:sp} 
There exists a factorization $A_T = U_T B_T^{\sf H}$ such that, for any $z \in \C - \support(\bnu)$, 
\begin{equation}
	\label{eq:ass6}
B_T^{\sf H} \left( R_T - z I_T \right)^{-1} B_T 
\toaslong 
m_\bnu(z) P 
\end{equation}
for some $P = \diag(p_1 I_{j_1},\ldots,p_t I_{j_t})$, $p_1 > \ldots > p_t > 0$, $j_1 + \ldots + j_t = K$
and where it is recalled that $m_\bnu(z)$ is the ST of the 
probability measure $\bnu$. 
\end{assumption}

\begin{remark}
	\label{rem:applications}
	Assumption~\ref{ass:sp} is in general very strong, \revju{as it 
requires in some sense that the right singular vectors of $A_T$ corresponding
to the non zero singular values show an isotropic behavior in the eigenbasis of
$R_T$. This condition is met in the following practical scenarios:}
	\begin{enumerate}
		\item {\bf Array Processing}: Let $A_T = H_T P^{1/2} S_T^{\sf H}$, with $H_T = [h(\theta_1), \cdots, h(\theta_K)]$ ($\theta_k$ distinct) the matrix of steering vectors, $P = \diag(a_1^2,\ldots,a_K^2)$ the source powers, $S_T = T^{-1/2} [ s_{t,k}^* ]_{t,k=1}^{T,K}$ the source signals, and let $V_T=W_TR_T^{1/2}$. Assume the $s_{k,t}$ i.i.d. of zero mean and unit variance and $[\sqrt{N}h(\theta)]_n=e^{-2\pi \imath n\sin(\theta)}$. Writing $A_T=U_TB_T^{\sf H}$ with $U_T=H_T(H_TH_T^{\sf H})^{-1/2}$ and $B_T=S_TP^{1/2}(H_TH_T^{\sf H})^{1/2}$, we can show \revju{$(H_TH_T^{\sf H})^{-1/2} \rightarrow I_K$} while \revjul{$S_T^\herm (R_T-zI_T)^{-1}S_T \asto m_\bnu(z)I_K$} so that Assumption~\ref{ass:sp} holds. See the proof of Lemma~\ref{lm:model} for details.

			\vspace{0.2cm}

		\item {\bf MIMO Communication}: Let $A_T = H_T P^{1/2} S_T^{\sf H}$, with $H_T=[h_1,\ldots,h_K]$ the wireless channels (i.i.d. zero mean $1/N$-variance entries) of $K$ transmitters, $P$ their diagonal power matrix and $S_T$ their matrix of transmitted (i.i.d. zero mean $1/T$-variance) signals. Taking $V_T=R_T^{1/2}W_T$, {\it i.e.} spatially correlated noise, and considering $Y_T^\herm$ instead of $Y_T$, we may write $A_T^\herm=U_TB_T^\herm$ with $U_T=S_T(S_TS_T^\herm)^{-1/2}$ and $B_T=H_T P^{1/2} (S_TS_T^\herm)^{1/2}$ to obtain $B_T^\herm (R_T-zI_N)^{-1}B_T\asto m_\bnu(z)P$.
	\end{enumerate}
\end{remark}

%Section~\ref{detect-spikes} will introduce the main results of the article, and in particular the new detection and estimation procedures under the general hypothesis of Assumption~\ref{ass:sp}. Since the results of Section~\ref{detect-spikes} may be difficult to grasp in the full generality of the set of hypotheses, we will then devote Section~\ref{example-doa} to the specific study of Item 1) in Remark~\ref{rem:applications} with $v_t$ an ARMA process for which an improved MUSIC algorithm to estimate the angles $\theta_k$ will be proposed.

\subsection{Results on the information-plus-noise matrix}

We recall here the main results concerning the eigenvalue distribution of $Y_T Y_T^{\sf H}$. Since $Y_T Y_T^{\sf H}$ is at most a rank $2K$ perturbation of $V_T V_T^{\sf H}$ with $K$ fixed, Weyl's interlacing inequalities \cite[Th. 4.3.6]{HornJoh'90} show, in conjunction with Theorem~\ref{th:lsm}, that the spectral measure of $Y_T Y_T^{\sf H}$ also converges to $\bmu$ in the sense of Theorem~\ref{th:lsm}-\ref{cvg-Q}). However, a finite number of eigenvalues of $Y_T Y_T^{\sf H}$ might stay isolated away from the support of $\bmu$ \cite[Th. 2.2]{ChapCouHac'12}: 
\begin{theorem}
\label{th:1stord}
Under Assumptions \ref{ass:c}--\ref{ass:sp}, let $\bmu$ and $[a,b]$ be as in Theorem~\ref{th:lsm}. Let $\hat\lambda_{1,T} \geq \cdots \geq \hat\lambda_{N,T}$ be the eigenvalues of $Y_T Y_T^{\sf H}$ with spectral measure $\hat\tau_T = N^{-1} \sum_{i=1}^N \hat\lambda_{i,T}$. Then: 
\begin{enumerate}
\item\label{same-msl}
For every bounded and continuous real function $f$,
\begin{equation*}
\int f(t) \hat\tau_T(dt) \ \toaslong \ \int f(t) \bmu(dt) . 
\end{equation*}
\item\label{neig-below}
For any interval $[x_1, x_2] \subset (0, a)$ 
\begin{equation*}
\sharp \{ i \, : \, \hat\lambda_{i,T} \in [x_1, x_2] \} 
= 0 \ \text{w.p.} \ 1 \ \text{for all large} \ T. 
\end{equation*}
\item\label{eig-above}
	The function $\bg(x) \triangleq x \bm(x) \tilde\bm(x)$ is positive and decreases from $\bg(b^+)$ to zero on $(b, \infty)$. If $p_1 \bg(b^+) \leq 1$, then $\hat\lambda_{1,T} \asto b$. Otherwise, let $s\in\{1,\ldots,t\}$ be the largest index for which $p_s \bg(b^+) > 1$. For $k=1,\ldots,s$, let ${\brho}_k$ be the unique solution $x$ in $(b,\infty)$ of $p_k \bg(x) = 1$. Then, for $i=1,\ldots,s$ and with $j_0=0$,
\begin{align*} 
\hat\lambda_{j_1+\cdots+j_{i-1}+1,T}, \ldots, 
\hat\lambda_{j_1+\cdots+j_{i},T} &\toaslong {\bs\rho}_i \\
\hat\lambda_{j_1+\cdots+j_{s}+1,T} &\toaslong b.
\end{align*} 

\item The condition $p_k \bg(b^+) > 1$ is equivalent to 
\begin{equation} 
\label{cond:spike} 
p_k > \left( \int \frac{-m_b}{1 + \bc m_b t} \bnu(dt) \right)^{-1} 
\end{equation} 
with $m_b$ the solution in $( -(\bc b_\bnu)^{-1}, 0)$ to Equation~\eqref{eq:m(b+)}. 
\end{enumerate} 
\end{theorem} 
\begin{proof}
The first two items in this theorem are proved in \cite{ChapCouHac'12} in a more general setting. To obtain the last item, observe that $\bg(x) = - \int \bm(x) (1+\bc\bm(x) t)^{-1} \bnu(dt)$ from the definition of $\tilde\bm$ in Theorem \ref{th:lsm}-\ref{st-coresolv}) and recall that $\bm(x) \downarrow m_b$ as $x\downarrow b$, where $m_b$ is defined in Corollary \ref{cor:edge}. 
\end{proof}

This theorem shows in particular that the number of isolated eigenvalues of $Y_T Y_T^\herm$ is upper bounded by the rank $K$ of $A_T$ and it reaches this rank if $p_t$ is large enough.

\begin{remark}
In the white noise setting, {\it i.e.} $R_T=I_T$ (hence, $\bnu = \delta_1$), $\bmu$ is the celebrated Marchenko-Pastur law, and Equation~\eqref{cond:spike} boils down to $p_k > \sqrt{\bc}$ (see {\it e.g.} \cite{HachLouMes'11}). The source detection approaches studied in \cite{penna-cl09, BianDebMai'11, nad-icc11} rely on this condition. 
\end{remark}

\section{Source detection and parameter estimation} 
\label{1st-ord} 

We start by stating the results in the general context of Assumptions~\ref{ass:c}--\ref{ass:sp}. We shall then deal more specifically with the model of Remark~\ref{rem:applications}-1). 

\subsection{General results} 
\label{detect-spikes} 

Theorem~\ref{th:1stord} gives the following signal dimension estimator: 
\begin{theorem}
\label{th:detection} 
Under Assumptions~\ref{ass:c}--\ref{ass:sp}, let $s\geq 0$ be the largest integer for which Equation~\eqref{cond:spike} holds. Let $0<\varepsilon<(\brho_s/b)-1$ with $\brho_0=\infty$. Given $L\geq K$, define
\begin{equation*}
\hat k_T = \arg\max_{k \in \{0,\ldots, L\}} \frac{\hat\lambda_{k,T}}{\hat\lambda_{k+1,T}} > 1 + \varepsilon 
\end{equation*}
with $\hat\lambda_{0,T}=\infty$. Then, for all $T$ large, w.p. 1,
\begin{equation*}
\hat k_T = j_1 + \ldots + j_s\quad (j_0=0).
\end{equation*}
\end{theorem}
\begin{proof}
	\revju{The result is clear for $s=0$. Else,} writing $\bs k = j_1 + \ldots + j_s$, Items \ref{same-msl}) and \ref{eig-above}) of Theorem~\ref{th:1stord} ensure $\hat\lambda_{\bs k,T} \asto \brho_s > b$ and $\hat\lambda_{\ell,T} \asto b$ for $\ell = \bs k +1, \ldots, L$.
\end{proof}

Theorem~\ref{th:detection} allows in practice to evaluate the number of strong sources when $T$ is large. This however requires $\varepsilon$ to be taken such that $\varepsilon<(\brho_s/b)-1$, a value which is practically not known. \revju{As the typical spacing between noise eigenvalues is of order $O(1/N)$ (see \emph{e.g.} \cite{pastur2011eigenvalue}), for all large $N$, one may take $\varepsilon$ such that $\varepsilon\to 0$ and $N\varepsilon \to \infty$ as $N\to\infty$.} Theorem~\ref{th:detection} also assumes that the receiver knows an upper bound $L$ on $K$, which \revju{is a common hypothesis}. 

In the sequel, for $i \in \{1, \ldots, K \}$, we let ${\cal K}(i) = 1$ if $1 \leq i \leq j_1$, ${\cal K}(i) = 2$ if $j_1+1 \leq i \leq j_1 + j_2$, $\ldots$, ${\cal K}(i) = t$ if $j_1+\cdots+j_{t-1}+1 \leq i \leq K$.  The following theorem provides a means for estimating consistently $p_1,\ldots,p_s$: 
\begin{theorem}
\label{th:estim-p} 
In the setting of Theorem \ref{th:detection}, let  
\begin{align*}
\hat m_T(x) &\triangleq \frac{1}{N-\hat k_T} \sum_{n=\hat k_T+1}^{N} \frac{1}{\hat\lambda_{n,T} - x} \\
\hat g_T(x) &\triangleq \hat m_T(x)( x c_T \hat m_T(x) + c_T - 1) \\
\hat p_{i,T}&\triangleq \frac{1}{\hat g_T(\hat\lambda_{i,T})},~i=1, \ldots, \hat k_T.
\end{align*}
Then 
\begin{equation*}
\hat p_{i,T} - p_{{\cal K}(i)} \toaslong 0.
\end{equation*}
\end{theorem} 
\begin{proof}
	Recall that $\lambda_{1,T}\geq\ldots\geq \lambda_{N,T}$ are the eigenvalues of $W_T R_T W_T^{\sf H}$. In the proof, we restrict the elementary events to belong to the probability one set where $\lambda_{1,T} \to b$, $\underline m_T(x) \to \bm(x)$ uniformly on the compact subsets of $( b, \infty)$ (see Theorem~\ref{th:lsm}-\ref{cvg-unif})), $\hat\lambda_{i,T} \to {\bs\rho}_{{\cal K}(i)}$ for $i=1,\ldots, j_1+\cdots+j_s$, $\hat\lambda_{j_1+\cdots+j_s+1,T} \to b$, and $\hat k_T \to j_1+\cdots+j_s$ (Theorems~\ref{th:lsm}--\ref{th:detection}). Observe that $Y_T Y_T^{\sf H}$ is at most a (nonnegative) rank $2K$ perturbation of $V_T V_T^\herm$. In these conditions, Weyl's inequalities \cite[Th.~4.3.6]{HornJoh'90} ensure $\hat\lambda_{n,T} \leq \lambda_{n-2K,T}$ and $\lambda_{n,T} \leq \hat\lambda_{n-2K,T}$ for $=2K+1,\ldots, N$. Then, for any $x > b$ and $T$ large,
\begin{align*} 
\hat m_T(x) %&= \frac{1}{N-\hat k_T} \left( \sum_{n=2K+1}^{N} \frac{1}{\hat\lambda_{n,T} - x} + \sum_{n=\hat k_T+1}^{2K} \frac{1}{\hat\lambda_{n,T} - x} \right) \\ 
&\geq \frac{1}{N - \hat k_T} \left( \sum_{n=1}^{N-2K} \frac{1}{\lambda_{n,T} - x} + \sum_{n=\hat k_T+1}^{2K} \frac{1}{\hat\lambda_{n,T} - x} \right) \\ 
&\triangleq \underline m_T(x) + e_T(x) 
\end{align*} 
where $e_T(x) \to 0$ uniformly on compact sets of $(b, \infty)$, and
\begin{align*} 
&\hat m_T(x)=\frac{1}{N-\hat k_T} \left( \sum_{n=\hat k_T+1}^{N-2K} \frac{1}{\hat\lambda_{n,T} - x}  + \sum_{n=N-2K+1}^N \frac{1}{\hat\lambda_{n,T} - x}  \right) \\ 
&\leq \frac{1}{N - \hat k_T} \left( 
\sum_{n=\hat k_T+1+2K}^{N} \frac{1}{\lambda_{n,T} - x} + \sum_{n=N-2K+1}^N \frac{1}{\hat\lambda_{n,T} - x}  \right) \\ 
&\triangleq \underline m_T(x) + e'_T(x) 
\end{align*} 
where $e'_T(x) \to 0$ uniformly on compact sets of $( b, \infty)$. Consequently, $\hat g_T(\hat\lambda_{i,T}) - \bg(\hat\lambda_{i,T}) \to 0$ for $i=1,\ldots,\hat k_T$. Clearly, $\bg(\hat\lambda_{i,T}) - \bg({\brho}_{{\cal K}(i)}) \to 0$ so that $\hat g_T(\hat\lambda_{i,T}) - \bg({\brho}_{{\cal K}(i)}) \to 0$ which, along with $\bg(\brho_{\mathcal K(i)})=1/p_{\mathcal K(i)}$, gives the result. 
\end{proof} 

Let now $A_T=U_TB_T^\herm$ following Assumption~\ref{ass:sp} and write $U_T=[U_{1,T},\ldots,U_{t,T}]$, $U_{\ell,T}\in\CC^{N\times j_\ell}$. We introduce the orthogonal projection matrix $\Pi_{\ell,T}=U_{\ell,T}U_{\ell,T}^\herm\in\CC^{N\times N}$. Similarly, we denote $\hat{\Pi}_{\ell,T}$ the orthogonal projection matrix on the eigenspace corresponding to the set of eigenvalues $\{\hat{\lambda}_{j_1+\ldots+j_{\ell-1}+1,T},\ldots, \hat{\lambda}_{j_1+\ldots+j_\ell}\}$ in $Y_TY_T^\herm$, for $\ell=1,\ldots,t$ ($j_0 = 0$). With these notations, we have the following estimate of bilinear forms of the type $a_{T}^{\sf H}{\Pi}_{\ell,T}b_{T}$: 
\begin{theorem}
	\label{th:quad_forms}
	Under Assumptions~\ref{ass:c}--\ref{ass:sp}, let $a_T,b_T\in\CC^N$ be two sequences of deterministic vectors with bounded norms and let $\mathcal K(i)\leq s$ with $s$ the largest integer for which \eqref{cond:spike} holds. Then:
\begin{align*}
	a_{T}^{\sf H}\Pi_{\mathcal K(i),T}b_{T} - \frac{\hat{g}_T'(\hat{\lambda}_{i,T})}{\hat{m}_T(\hat{\lambda}_{i,T})\hat{g}_T(\hat{\lambda}_{i,T})} a_{T}^{\sf H}\hat{\Pi}_{\mathcal K(i),T}b_{T}\toaslong 0.
\end{align*}
\end{theorem}
\begin{proof}
	From Assumption~\ref{ass:sp}, $B_T^\herm B_T\asto P$ (multiply each side of \eqref{eq:ass6} by $-z$ and take $z$ large). Therefore, $p_1,\ldots,p_t$ are the limiting positive eigenvalues of $A_TA_T^\herm$. For $R_T=I_N$, the theorem thus coincides with \cite[Theorem~2]{HachLouMes'11} since then $V_T=W_T$ is a bi-unitarily invariant (here Gaussian) matrix as requested by \cite[Assumption~2]{HachLouMes'11}. We now reproduce the steps of \cite[Theorem~2]{HachLouMes'11} under our set of assumptions. \cite[Equation~(8)]{HachLouMes'11} remains valid in our setting which, under the present notations, reads
	\begin{align}
		\label{eq:quad_form}
		a_{T}^{\sf H}\hat{\Pi}_{\ell,T}b_{T} &= -\frac{1}{\imath \pi}\oint_{\mathcal C_{\ell,T}} \tilde{a}_{T}^{\sf H}\underline{Q}_T(z)\tilde{b}_T dz 
		+ \frac{1}{\imath \pi}\oint_{\mathcal C_{\ell,T}} \hat{a}_{T}^{\sf H}\hat{H}_T(z)^{-1}\hat{b}_{T}dz
	\end{align}
for $\mathcal C_{\ell,T}$ a complex positively oriented contour enclosing only the eigenvalues $\hat{\lambda}_{j_1+\ldots+j_{\ell-1}+1,T},\ldots,\hat{\lambda}_{j_1+\ldots+j_\ell,T}$, with
	\begin{align*}
		\tilde{a}_{T}^{\sf T} &= [a_{k,T}^{\sf T}, 0,\ldots,0],~\tilde{b}_{T}^{\sf T} = [b_{k,T}^{\sf T}, 0,\ldots,0] \\
		Q_T(z) &= (V_TV_T^\herm - z I_N)^{-1},~\tilde{Q}_T(z)=(V_T^\herm V_T - z I_T)^{-1} \\
		\underline Q_T(z) &= \begin{bmatrix} zQ_T(z^2) & V_T\tilde{Q}_T(z^2) \\ \tilde{Q}_T(z^2)V_T^\herm & z\tilde{Q}_T(z^2) \end{bmatrix}
	\end{align*}
	\begin{align*}
			\hat{a}_{T} &= \begin{bmatrix} \revju{zU_T^\herm Q_T(z^2)} \\ B_T^\herm \tilde{Q}_T(z^2)V_T^\herm \end{bmatrix} a_T,~ \hat{b}_T = \begin{bmatrix} \revju{zU_T^\herm Q_T(z^2)} \\ B_T^\herm\tilde{Q}_T(z^2)V_T^\herm \end{bmatrix} b_T \\
					\hat{H}_{T}(z) &= \begin{bmatrix} zU_T^\herm Q_T(z^2)U_T & U_T^\herm V_T\tilde{Q}_T(z^2)B_T + I_K \\ B_T^\herm\tilde{Q}_T(z^2)V_T^\herm U_T + I_K & zB_T^\herm\tilde{Q}_T(z^2)B_T \end{bmatrix}.
	\end{align*}
	Let $\ell\leq s$. From Theorem~\ref{th:1stord}-2), for all large $T$ w.p. 1, the first term on the right-hand side of \eqref{eq:quad_form} is null (no pole of $\underline{Q}_T$ lies in $\mathcal C_{\ell,T}$ for large $T$), while in the second term $\mathcal C_{\ell,T}$ can be replaced by a contour $\mathcal C_\ell$ enclosing ${\brho}_\ell$ but no $\brho_k$, $k\neq \ell$. We must now prove $\hat{a}_T^\herm\hat{H}_T(z)\hat{b}_T- \bar{a}_T^\herm \bar{H}_T(z)\bar{b}_T\asto 0$ where
	\begin{align*}
		\bar{a}_T &= \begin{bmatrix} z\bm(z^2)U_T^\herm \\ 0 \end{bmatrix} a_T ,~ \bar{b}_T = \begin{bmatrix} z\bm(z^2)U_T^\herm \\ 0 \end{bmatrix} b_T \\
			\bar{H}_T(z) &= \begin{bmatrix} z\bm(z^2)I_K & I_K \\ I_K & z\tilde{\bm}(z^2)P \end{bmatrix}.
	\end{align*}
	By \cite[Lemmas~4.1--4.6]{ChapCouHac'12}, $\Vert \hat{a}_T - \bar{a}_T \Vert \asto 0$, $\Vert \hat{b}_T - \bar{b}_T \Vert \asto 0$,
	\begin{align*}
		\left\Vert \hat{H}_T(z) - \begin{bmatrix} z\bm(z^2)I_K & I_K \\ I_K & \frac{B_T^\herm \left( I_T + \bc {\bm}(z^2)R_T \right)^{-1}B_T}{-z} \end{bmatrix} \right\Vert \asto 0.
	\end{align*}
	Assumption~\ref{ass:sp} and the definition of $\tilde{\bm}(z)$ then imply $\Vert \frac{-1}{z} B_T^\herm \left( I_T + \bc {\bm}(z^2)R_T \right)^{-1}B_T - z\tilde{\bm}(z^2)P\Vert \asto 0$, which finally gives $\hat{a}_T^\herm\hat{H}_T(z)\hat{b}_T- \bar{a}_T^\herm \bar{H}_T(z)\bar{b}_T\asto 0$. For $z\in\mathcal C_\ell$, $z\bm(z^2)$ and $z\tilde{\bm}(z^2)$ are bounded by $[{\bf d}(\mathcal C_\ell,{\rm supp}(\bmu))]^{-1}$. Take $0<\varepsilon<{\bf d}(\mathcal C_\ell,{\rm supp}(\bmu))$. Then, for all large $T$, $zQ_T(z^2)$ and $z\tilde{Q}(z^2)$ are bounded by $\varepsilon^{-1}$ w.p. 1. The dominated convergence theorem therefore ensures that
	\begin{align*}
		a_{T}^{\sf H}\hat{\Pi}_{\ell,T}b_{T} - \frac1{\imath \pi}\oint_{\mathcal C_\ell} \bar{a}_T^\herm \bar{H}_T(z)^{-1}\bar{b}_T dz \asto 0.
	\end{align*}
	Residue calculus of the right-hand side integrand as in \cite[Equations~(10)-(11)]{HachLouMes'11} then gives
\begin{align*}
	a_{T}^{\sf H}\hat{\Pi}_{\ell,T}b_{T}-\frac{\bm({\brho}_\ell){\bg}({\brho}_\ell)}{\bg'(\brho_\ell)}a_{T}^{\sf H}\Pi_{\ell,T}b_{T}\toaslong 0.
\end{align*}
Take $i$ such that $\mathcal K(i)=\ell$. Using $\hat{\lambda}_{i,T}\asto \brho_{\ell}$, $\hat{m}_T(x)\asto \bm(x)$, $\hat{g}_T(x)\asto \bg(x)$, and $\hat{g}'_T(x)\asto \bg'(x)$ for $x$ outside the support of $\bmu$ then concludes the proof.
\end{proof}

\subsection{Narrowband array processing} 
\label{example-doa} 

We now apply the results of Section~\ref{detect-spikes} to the array processing model of Remark~\ref{rem:applications}. Consider a uniform linear array of $N$ antennas which captures $T$ successive realizations $y_1,\ldots,y_T$ of the random process:
\begin{equation} 
\label{input_output0}
y_t = \sum_{k=1}^K a_k h(\theta_k)s_{k,t} + v_t
\end{equation} 
with $a_1\geq \ldots\geq a_K > 0$ the amplitude of sources $1,\ldots,K$, $h(\theta)\in\CC^N$ the steering-vector function
\begin{equation}
h(\theta)=\frac{1}{\sqrt{N}} \left[1, e^{-2\imath \pi \sin \theta}, \ldots, e^{-2\imath \pi (N-1) \sin \theta} \right]^{\sf T} \label{a_i}
\end{equation}
with $\theta_k$ the angle-of-arrival of the signal from source $k$ (the $\theta_k$ are assumed distinct), $s_{k,t}\in\CC$ the signal emitted by source $k$ at time $t$ such that $(s_{t,k})_{t,k=1}^{\infty, K}$ is an infinite array of circular complex i.i.d. random variables with $\E s_{1,1} = 0$, $\E | s_{1,1} |^2 = 1$, and $\E | s_{1,1} |^8 < \infty$, and $v_t\in\CC^N$ the noise received at the sensor array at time $t$.

Denoting \revju{$Y_T=T^{-1/2}[y_1,\ldots,y_T]\in\CC^{N\times T}$}, \eqref{input_output0} reads
\begin{equation}
\label{model}
Y_T = H_T P^{1/2} S_T^{\sf H} + V_T 
\end{equation}
where $H_T=\left[h(\theta_1), h(\theta_2), \ldots, h(\theta_K) \right]\in\CC^{N\times K}$, $S_T = T^{-1/2} [ s_{t,k}^* ]_{t,k=1}^{T,K}\in\CC^{T\times K}$, $P= \diag(a_1^2, \ldots, a_K^2)$, and $V_T=T^{-1/2} [v_1,\ldots,v_T]\in\CC^{N\times T}$. We assume the rows of $\sqrt{T} V_T$ to be independent snapshots of a complex Gaussian circular causal $\text{ARMA}(m,n)$ stationary process. This process can be represented as the output of a filter with transfer function $\bp(z) = (1 + \alpha_1 z^{-1} + \ldots + \alpha_m z^{-m} ) / (1 + \beta_1 z^{-1} + \ldots + \beta_n z^{-n} )$ driven by a standard complex Gaussian circular white noise. For $|z| \geq 1$, $\bp(z) = \sum_{\ell=0}^\infty \psi_\ell z^{-\ell}$ where $\sum |\psi_\ell| < \infty$, and we can write $V_T = W_T R_T^{1/2}$ with $W_T$ as in Assumption~\ref{ass:gauss} and 
\begin{equation*}
R_T =\begin{bmatrix}
r_0 & r_{1} & \ldots & r_{T-1} \\ 
r_{-1} & \ddots & \ddots & \vdots \\ 
\vdots  & \ddots & \ddots & r_{1}\\ 
r_{1-T} & \ldots & r_{-1} & r_0 
\end{bmatrix} 
\end{equation*}
with $r_k = \sum_{\ell \geq 0} \psi_{\ell+k} \psi_\ell^*$ for any $k\in \N$, the matrix being nonnegative.

% The assumptions about circularity and the upper moments will be used only 
% in the second order analysis. The circularity assumptions have been added 
% only for simplicity. Second order analysis can be generalized to the case 
% where the $s_{t,k}$ are not circular. 

\begin{lemma}
\label{lm:model}
Under Assumption~\ref{ass:c}, the model \eqref{model} satisfies Assumptions~\ref{ass:gauss}--\ref{ass:sp} with $\bnu$ defined by
\begin{equation} 
\label{def-nu} 
\int g(t) \bnu(dt) = \int_0^1 g(|\bp(\exp(2\imath \pi u))|^2) \, du 
\end{equation} 
for every positive measurable function $g$, and with $P$ in Assumption~\ref{ass:sp} the matrix of the source powers $a_k^2$.
\end{lemma} 
\begin{proof}
	We start with Assumptions~\ref{ass:R} and \ref{nu:edge}. If $m=n=0$, then $\bnu = \delta_1$ and these assumptions are trivially satisfied. Assume $\revju{\max}(m,n) > 0$. Then Assumption~\ref{ass:R}--\ref{R:msl}) is a well known result on the spectral behavior of large Toeplitz matrices \cite{Gray'06, GrenSze'84}. The support of $\bnu$ is the compact \revju{non-singleton} interval $[a_\bnu, b_\bnu] = [ \min_u q(u), \max_u q(u) ]$, $q(u) \triangleq |\bp(\exp(2\imath \pi u))|^2$.  It is also well known \cite[\S 4.2]{Gray'06} that $a_\bnu \leq \sigma^2_{t,T} \leq b_\bnu$, so that Assumption \ref{ass:R}--\ref{R:noeig}) is satisfied. Since $\bp(z)$ is ARMA, for $g(t)$ the indicator function on a set of Lebesgue measure zero, the right hand side of \eqref{def-nu} is zero. Hence $\bnu$ has a density $f_\bnu$ with respect to the Lebesgue measure. Let us provide the expression of $f_\bnu$ at a point $s \in (a_\bnu, b_\bnu)$ such that for any $u$ for which $q(u) = s$, $q'(u) \neq 0$. In a neighborhood of any of these $u$, $q$ has a local inverse that we denote $q^{(-1)}_u$. Then, for $\varepsilon > 0$ small enough,
\begin{align*} 
\bnu(s-\varepsilon, s+\varepsilon) 
&= 
\int_{t \, : \, q(t) \in [s-\varepsilon, s+\varepsilon]} dt = 
\sum_{u \, : \, q(u) = s} \int_{[s-\varepsilon, s+\varepsilon]} 
\frac{1}{\Bigl| q'(q^{(-1)}_u(v)) \Bigr|} \, dv 
\end{align*} 
by the variable change $q(t) = v$. Letting $\varepsilon \downarrow 0$, 
we obtain
\begin{equation*}
\lim_{\varepsilon\downarrow 0} 
\frac{\bnu(s-\varepsilon, s+\varepsilon)}{2\varepsilon} = 
\sum_{u \, : \, q(u) = s} 
\frac{1}{| q'(u) |} 
= f_\bnu(s) . 
\end{equation*}
This proves $f_\bnu(s) \to \infty$ as $s \uparrow b_\bnu$, implying Assumption~\ref{nu:edge}. \\
We now turn to Assumptions~\ref{ass:A} and \ref{ass:sp}. Since the $\theta_i$ are distinct (modulo $\pi$), $H_T^{\sf H} H_T \to I_K$. By the law of large numbers, $S_T^{\sf H} S_T \toaslong I_K$. Hence $\rank(A_T) = K$ w.p.~1 for
all large $T$, and $\sup_T \| A_T \| < \infty$ w.p.~1. Let us write $A_T = U_T B_T^{\sf H}$ where $U_T = H_T (H_T^{\sf H} H_T)^{-1/2}$ and where $B_T = S_T P^{1/2} (H_T^{\sf H} H_T)^{1/2}$. By \cite[Lemma 2.7]{BaiSil'98} and $\E|s_{1,1}|^8<\infty$, for any $z \in \C_+$ and any $1\leq i,j\leq K$, 
\begin{equation*} 
\E \Bigl| \Bigl[S_T^\herm ( R_T - z I_T)^{-1} S_T - \frac{\tr[ ( R_T - z I_T)^{-1}]}{T} I_K \Bigr]_{i,j} \Bigr|^4 \leq \frac{C}{T^2} 
\end{equation*} 
for some $C > 0$. By Markov's inequality, the argument of $\E| \cdot |^4$ converges to zero w.p. 1, and this convergence can be extended to $\C - \support(\bmu)$. Since $T^{-1} \tr[ ( R_T - z I_T)^{-1}] \to m_{\bnu}(z)$ for $z \in \C - \support(\bnu)$, Assumption~\ref{ass:sp} is satisfied. 
\end{proof}

With these results, Lemma~\ref{lm:model} and Theorems~\ref{th:detection} and \ref{th:estim-p} lead to the following inference methods: 
\begin{proposition} 
\label{1stord-M} 
Consider the model \eqref{model}. Let $k\geq 0$ be the largest integer for which (take $a_0=\infty$) 
\begin{equation}
\label{detectability} 
a^2_k > \left( \int_0^1 \frac{-m_b}{1 + \bc m_b 
\, |\bp(\exp(2\imath \pi u))|^2} \, du \right)^{-1} 
\end{equation} 
with $m_b\in( - (\bc \max_u |\bp(\exp(2\imath \pi u))|^2 )^{-1}, 0)$ the solution of
\begin{equation*}
\int_0^1 \left( \frac{m \, |\bp(\exp(2\imath \pi u))|^2}
{1+\bc m \, |\bp(\exp(2\imath \pi u))|^2} \right)^2 du \ = \ 
\frac{1}{\bc}.
\end{equation*}
Given $L \geq K$ and $\varepsilon>0$, define (with $\hat\lambda_{0,T}=\infty$)
\begin{equation*}
\hat k_T = \arg\max_{m \in \{0,\ldots, L\}} 
\frac{\hat\lambda_{m,T}}{\hat\lambda_{m+1,T}} > 1 + \varepsilon . 
\end{equation*}
Then $\hat k_T = k$ w.p. 1 for all large $T$ and $\varepsilon$ small enough. Moreover, for $i=1,\ldots, \hat k_T$ let  $\hat a^2_{i,T} \triangleq (\hat g_T(\hat\lambda_{i,T}))^{-1}$ with $\hat g_T(\hat\lambda_{i,T})$ as in Theorem~\ref{th:estim-p}. Then 
\begin{equation*}
\hat a^2_{i,T} \asto a^2_{i}.
\end{equation*}
\end{proposition} 

From Theorem~\ref{th:quad_forms}, we now provide a source localization method based on MUSIC \cite{Schm'86}. Recall that MUSIC exploits the fact that $h(\theta_i)^\herm (I_N - \Pi_{1,T}^\ell)h(\theta_i)=0$ with $\Pi_{1,T}^\ell$ a projector on the subspace generated by $h(\theta_1),\ldots,h(\theta_\ell)$ for any $i\leq \ell\leq K$. Since $\Vert h(\theta)\Vert =1$, $\theta_1,\ldots,\theta_\ell$ are the arguments of the local maxima of 
\begin{equation*}
\gamma^\ell_T(\theta)\triangleq h(\theta)^\herm\Pi_{1,T}^\ell h(\theta).
\end{equation*}

\begin{proposition}
	\label{prop:GMUSIC}
	Let $k$ and $\hat{k}_T$ be as in Proposition~\ref{1stord-M} and denote $\hat{u}_{1,T},\ldots,\hat{u}_{\hat{k}_T,T}$ the eigenvectors of $Y_TY_T^\herm$ with respective eigenvalues $\hat{\lambda}_{1,T},\ldots,\hat{\lambda}_{\hat{k}_T,T}$. Then, for $\theta\in[-\pi/2,\pi/2]$,
\begin{align*}
	\gamma^{k}_T(\theta) - \hat{\gamma}^{\hat{k}_T}_T(\theta) \asto 0
\end{align*}
where 
\begin{align*}
	\gamma^k_T(\theta) &\triangleq h(\theta)^\herm \Pi_{1,T}^k h(\theta) \\
	\hat{\gamma}^{\hat{k}_T}_T(\theta) &\triangleq \sum_{j=1}^{\hat{k}_T} \frac{\hat{g}_T'(\hat{\lambda}_{j,T})}{\hat{m}_T(\hat{\lambda}_{j,T})\hat{g}_T(\hat{\lambda}_{j,T})} h(\theta)^\herm \hat{u}_{j,T}\hat{u}_{j,T}^\herm h(\theta).
\end{align*}
\end{proposition}
\begin{proof}
	Lemma~\ref{lm:model} ensures that Assumptions~\ref{ass:c}--\ref{ass:sp} are satisfied, so Theorem~\ref{th:quad_forms} can be applied for each $i\leq k$. Taking $a_T=b_T=h(\theta)$ and $U_T=H_T(H_T^\herm H_T)^{-1/2}$ as in Theorem~\ref{th:quad_forms}, we obtain the desired result for $U_T J U_T^\herm$, $J=\diag(I_{k},0)$, instead of $\Pi_{1,T}^k$. As $(H_T^\herm H_T)^{-1/2}J(H_T^\herm H_T)^{-1/2}\to J$ and $h(\theta)^\herm(H_TJH_T^\herm-\Pi_{1,T}^k)h(\theta)\to 0$, we have $h(\theta)^\herm\Pi_{1,T}^k h(\theta)-h(\theta)^\herm U_TJ U_T^\herm h(\theta)\to 0$, completing the proof.
\end{proof}

Proposition~\ref{prop:GMUSIC} ensures that $\hat{\gamma}^{\hat{k}_T}_T(\theta)$ is a consistent estimator of the localization function $\gamma^k_T(\theta)$. The \revju{improved} MUSIC algorithm we therefore propose consists in estimating $\theta_1,\ldots,\theta_k$ as the arguments of the $\hat{k}_T$ highest maxima of $\hat{\gamma}^{\hat{k}_T}_T(\theta)$. Observe that, although the system models differ in both articles, the MUSIC estimator proposed here takes the same form as that provided in \cite{HachLouMes'11}. This remark would not hold if it were not for Assumption~\ref{ass:sp}.

\revju{Note also that, as $\bc\to 0$, $\bg'(x)\bm(x)^{-1}\bg(x)^{-1}\to 1$ for all real $x\neq \int t \bnu(dt)$, so that the improved MUSIC algorithm proposed reduces to the standard large $T$ MUSIC approach.}

\section{Second order performance analysis} 
\label{2nd-ord}
In this section, we discuss the asymptotic (second order) performance of the detection and estimation schemes derived in Section~\ref{1st-ord}. The model of Section~\ref{example-doa} is considered. Following the notations of Section~\ref{detect-spikes}, we gather the source powers $a_k^2$ in groups of equal powers $p_1>...>p_t$ with respective multiplicities $j_1,\ldots,j_t$.
\subsection{Main results}
We start by studying the fluctuations of the isolated eigenvalues of $Y_T Y_T^{\sf H}$. Recall the definition of $\nu_T$ in Assumption~\ref{ass:R} and recall that $c_T = N/T$. Replacing $\bnu$ and $\bc$ with $\nu_T$ and $c_T$, respectively, in Theorem~\ref{th:lsm}, we obtain that
\begin{equation}
m_T(z) = \left( -z + \int \frac{t}{1 + c_T m_T(z) t} \nu_T(dt) \right)^{-1}
\label{m_T}
\end{equation}
uniquely defines the ST $m_T(z)$ of a probability measure $\mu_T$ supported by $\R_+$. In addition, $\mu_T$ converges weakly to $\bmu$ as $T \to\infty$; the Hausdorff distance between the supports of these two measures converges to zero \cite{SilvBai'95,BaiSil'98} and, for each $b' > b$, $m_T(z)$ is analytic on $\C - [0,b']$ for all large $T$.
% Let $m_T(z)$ be the ST defined analogously to $\bm(z)$ 
% in Theorem \ref{th:lsm} with the difference that the measure $\nu$ and the 
% constant $\bc$ are repectively replaced with their ``finite horizon'' 
% equivalents $\nu_T$ (see Assumption \ref{ass:P}) and $c_T = N/T$. Let 
Let 
\begin{align*}
\tilde m_T(z) &= 
\int\frac{-1}{z(1 + c_T m_T(z) t)} \nu_T(dt) = 
\frac{-1}{zT} \tr ( I_T + c_T m_T(z) R_T )^{-1} .
\end{align*} 

Similarly to Theorem \ref{th:lsm}-\ref{st-coresolv}), $\tilde m_T(z)$ satisfies $\tilde m_T(z) = c_T m_T(z) - (1 - c_T) / z$.
% By the well known properties of 
% convergence of $m_T(z)$ towards $\bm(z)$ (see \emph{{\it e.g.}} \cite{SilvBai'95}), 
Consequently, for all $T$ large, $g_T(x) \triangleq x m_T(x) \tilde m_T(x)$ is defined on $(b',\infty)$, $b' > b$, and, for any $k$ such that $p_k \bg(b^+) > 1$, $p_k g_T(x) = 1$ has a unique solution $\rho_{k,T}$ in $(b, \infty)$. 

The main result of this section (Theorem~\ref{th-fluct-spikes}) describes the fluctuations of $\hat\lambda_{i,T} - \rho_{{\cal K}(i),T}$, $i\leq s$, with $s$ the largest integer satisfying \eqref{cond:spike}. \revju{We start by introducing the important quantity ${\bDelta}(x)$.}
\begin{lemma}
\label{lmDelta}
Consider the model \eqref{model}. Then the function
\begin{equation*}
{\bDelta}(x) = 
1 - \bc \int \left(\frac{\bm(x) t}
{1+\bc\bm(x) t}\right)^2 \bnu(dt)
\end{equation*}
is defined and positive on $(b, \infty)$. Furthermore, ${\bDelta}(x) \to 0$
as $x \downarrow b$ and ${\bDelta}(x) \to 1$ as $x\to\infty$. 
\end{lemma} 
\begin{proof}
	See Appendix~\ref{prf-Delta}.
\end{proof}

\begin{theorem}
\label{th-fluct-spikes}
Consider \eqref{model} with the assumptions of Section~\ref{example-doa}. Assume in addition $\E [s_{1,1}^u (s_{1,1}^*)^v] = 0$ for $u+v \leq 4$ and $u\neq v$, and let $\kappa \triangleq \E|s_{1,1}|^4 - 2$. Let $s$ be the largest integer (assumed $\geq 1$) for which \eqref{detectability} holds. For $k=1,\ldots, s$ and all $T$ large, let $\rho_{k,T}$ be the unique solution in $(b, \infty)$ of $p_k g_T(x) = 1$. Define (with $j_0=0$)
\begin{equation*}
\eta_{k,T} = \sqrt{T} \left( 
\begin{bmatrix} \hat\lambda_{j_1+\cdots+j_{k-1}+1,T} \\ 
\vdots \\
\hat\lambda_{j_1+\cdots+j_{k},T} 
\end{bmatrix} 
- 
\rho_{k,T} \begin{bmatrix} 1 \\ \vdots \\ 1 \end{bmatrix} 
\right),
\end{equation*}
\begin{align*} 
\alpha_k &= 
\frac{\bm^2(\brho_k)}{{\bDelta}(\brho_k)}\left[
\int \frac{t^2+2p_k t}{(1+\bc\bm(\brho_k)t)^2}\bnu(dt) \right. \\
& \left. \phantom{=\frac{\bm^2(\brho_k)}{ {\bDelta}(\brho_k)}}
+ \bc \Bigl( \int \frac{p_k\bm(\brho_k)t}
{(1+\bc\bm(\brho_k)t)^2}\bnu(dt) \Bigr)^2 \right],  \\ 
\beta_k &= \int \frac{p_k^2 \bm(\brho_k)^2}
{(1 + \bc \bm(\brho_k) t)^2} \bnu(dt), \quad \text{and} \\ 
\phi_k &= \left( \int \frac{p_k \bm(\brho_k)}
{1 + \bc \bm(\brho_k) t} \bnu(dt) \right)^2 .
\end{align*}
Let $M_1, \ldots, M_s$, $M_k = [ M_{\ell,m,k} ]_{1\leq\ell,m\leq j_k}$, be random independent Hermitian matrices such that $\{ M_{\ell,m,k} \}_{\ell\leq m}$ are independent, $M_{\ell,\ell,k} \sim {\mathcal N}(0, \alpha_k+\beta_k+\kappa\phi_k)$, and $M_{\ell,m,k} \sim \mathcal{CN}(0, \alpha_k+\beta_k)$ for $1\leq \ell < m \leq j_k$. Let $\chi_k$ be the $\R^{j_k}-$valued vector of the decreasingly ordered eigenvalues of $(p_k \bg'(\brho_k))^{-1} M_k$. Then
\begin{align*}
(\eta_{1,T}, \ldots, \eta_{s,T})\tolawlong (\chi_1,\ldots, \chi_s).
\end{align*}
\end{theorem} 
\begin{proof}
	The proof is provided in Section~\ref{prf-fluct-spikes}.
\end{proof}

Theorem~\ref{th-fluct-spikes} shows that, after appropriate centering and scaling, the vector of the isolated eigenvalues of $Y_T Y_T^{\sf H}$ that converge to $\brho_k > b$ tends to fluctuate like the eigenvalues of a certain Hermitian matrix with Gaussian elements. If $\kappa=0$, this matrix is a scaled Gaussian Unitary Ensemble (GUE) matrix.\footnote{We recall that a GUE matrix is a random Hermitian matrix $M = [M_{ij}]$ such that $M_{ii} \sim {\cal N}(0,1)$ and $M_{ij} \sim {\cal CN}(0,1)$ for $i < j$, these random variables being independent.} When $K = 0$, ${\bs s} T^{2/3} ( \hat\lambda_{1,T} - b_T )$ converges in law to the Tracy-Widom probability distribution ${\sf TW}(\cdot)$, where $b_T$ is the finite horizon equivalent to $b$ and $\bs s$ is a scaling parameter that depends on $\bc$ and $\bnu$ \cite{elkar'07}. This result can be generalized to show that for any fixed integer $r$, the vector $T^{2/3} ( \hat\lambda_{1,T} - b_T, \ldots, \hat\lambda_{r,T}-b_T)$ converges in distribution to a multidimensional version of the Tracy-Widom law. These results and Theorem~\ref{th-fluct-spikes} can then be used to evaluate the error probabilities of the source detection schemes described in Theorem~\ref{th:detection} and Proposition~\ref{1stord-M}. 
\begin{remark}
	We note without proof that for the specific ARMA model considered here, the measure $\nu_T$ can be freely replaced with $\bnu$ in Equation (\ref{m_T}), \revju{which arises from the fact that $\sqrt{T}(\nu_T-\bnu)\tolawshort 0$}. The error incurred on $m_T(z)$ by this replacement is negligible in the ARMA context. 
\end{remark}
\revju{From Theorem \ref{th-fluct-spikes}, one then retrieves the fluctuations of the source power estimates:}
\begin{theorem}
\label{2ord-powers}
Consider the setup of Theorem~\ref{th-fluct-spikes} and let $\hat p_{i,T}=(\hat g_T(\hat\lambda_{i,T}))^{-1}$ for $i=1,\ldots, j_1 + \cdots + j_s$. For $k=1,\ldots, s$, define (with $j_0=0$)
\begin{equation*}
\xi_{k,T} = \sqrt{T} \left( 
\begin{bmatrix} \hat p_{j_1+\cdots+j_{k-1}+1,T} \\ 
\vdots \\
\hat p_{j_1+\cdots+j_{k},T} 
\end{bmatrix} 
- 
p_k \begin{bmatrix} 1 \\ \vdots \\ 1 \end{bmatrix} 
\right).
\end{equation*}
Let $M_k$ be defined as in Theorem~\ref{th-fluct-spikes} and let $\check{\chi}_k$ be the $\R^{j_k}-$valued vector of the decreasingly ordered
eigenvalues of $p_k M_k$. Then
\begin{equation*}
(\xi_{1,T}, \ldots, \xi_{s,T}) \tolawlong (\check{\chi}_1,\ldots, \check{\chi}_s).
\end{equation*}
\end{theorem}
\begin{proof}
A sketch of the proof is given in Appendix~\ref{prf:2ord-powers}.
\end{proof}

\revju{A straightforward application of the Delta method \cite[Th. 3.1]{VAN00} on Theorem \ref{2ord-powers} implies in particular that, for $k=1,\ldots,s$,
\begin{align*}
	\sqrt{T} \left( \frac1{j_k} \sum_{i=1}^{j_k} \hat{p}_{j_1+\ldots+j_{k-1}+i,T} - p_k \right) \tolawlong \bar{\chi}_k
\end{align*}
with $\bar\chi_k \sim \mathcal{N}(0,j_k^{-1}p_k^2(\alpha_k+\beta_k+\kappa \phi_k))$, independent across $k$.
}
As a corollary of Theorem~\ref{2ord-powers}, the following proposition provides the behavior of the power estimates for extreme values of $p_k$, {\it i.e.} for $p_k\to\infty$ and for $p_k$ close to the detectability limit given by \eqref{detectability}:
\begin{proposition}
\label{limit_snr}
Consider the setting of Theorem~\ref{2ord-powers}. Let $p_{\rm lim}$ be the infimum of the $p_k$ satisfying \eqref{detectability}, $M_k$ be defined as in Theorem~\ref{th-fluct-spikes}, and $\psi_k \triangleq \alpha_k + \beta_k + \kappa\phi_k$, $\breve{\psi}_k \triangleq \alpha_k + \beta_k$. Then
\begin{align*}
\psi_k \xrightarrow[p_k\downarrow p_{\rm lim}]{} \infty&, \quad \breve{\psi}_k \xrightarrow[p_k \downarrow p_{\rm lim}]{} \infty \\
\psi_k \xrightarrow[p_k\to \infty]{} 1+\kappa&, \quad \breve{\psi}_k \xrightarrow[p_k\to \infty]{} 1.
\end{align*}
\end{proposition}
\begin{proof}
	See Appendix~\ref{prf-limsnr}.
\end{proof}

\subsection{Proof of Theorem \ref{th-fluct-spikes}} 
\label{prf-fluct-spikes}

The proof relies on two ingredients: an adaptation of \cite[Th.
2.3]{ChapCouHac'12} and a result on fluctuations of quadratic forms. Let
$A_T = U_T B_T^{\sf H}$ with $U_T = H_T (H_T^{\sf H} H_T)^{-1/2}$ and $B_T
= S_T P^{1/2} (H_T^{\sf H} H_T)^{1/2}= [ B_{1,T}, \ldots, B_{t,T} ]$,
$B_{k,T}\in\CC^{T \times j_k}$. In \cite{ChapCouHac'12}, it is shown that the
$\eta_{k,T}$ fluctuate like the ordered eigenvalues of the matrices $(p_k
\bg(\brho_k)')^{-1} ( \sqrt{\alpha_k} G_k + \sqrt{T} F_{k,T} )$ where $F_{k,T}
= m_T(\rho_{k,T}) B_{k,T}^{\sf H} ( I_T + c_T m_T(\rho_{k,T}) R_T )^{-1}
B_{k,T} + I_{j_k}$ and the $G_k$ are GUE matrices independent of the $F_{k,T}$.
\revju{This is formalized by Proposition \ref{chapon} below.}  
Using $H_T^{\sf H} H_T \asto I_K$, the law of large numbers and the definition
of $\rho_{k,T}$ informally give
\begin{align*} 
F_{k,T} 
&\simeq 
\left( \frac{p_k}{T} \tr\left[ m_T(\rho_{k,T}) ( I_T + c_T m_T(\rho_{k,T}) R_T)^{-1} \right] +1 \right) I_{j_k} = 0.
\end{align*}
We thus need to study the fluctuations of $\sqrt{T} F_{k,T}$, which is the
purpose of the \revju{three} following lemmas.
\revju{Lemma \ref{tlc:fq} is a Central Limit Theorem characterizing the 
fluctuations of random matrices of the type $S_T^{\sf H} D_T S_T$ where $D_T$ 
is a sequence of $T\times T$ deterministic matrices. Lemma 
\ref{lm-diag} particularizes the results of Lemma \ref{tlc:fq} to the case 
where $D_{T} = p_k m_T(\rho_{k,T}) ( I_T + c_T m_T(\rho_{k,T}) R_T)^{-1}$. 
In Lemma \ref{lm:fq-totale} these results are used to characterize the 
fluctuations of $F_{k,T}$. Essentially, it is shown there that the matrices
$B_{k,T}$ can be replaced with $\sqrt{p_k} S_{k,T}$. 
Lemmas \ref{tlc:fq}--\ref{lm:fq-totale} are proved in 
Appendices~\ref{prf-tlc:fq}--\ref{prf-totale} respectively:} 
\begin{lemma}
\label{tlc:fq}
Let $D_T\in\CC^{T\times T}$ be a sequence of deterministic Hermitian matrices with $\sup_T \| D_T \| < \infty$. Assume that
\begin{equation*}
\frac 1T \tr D_T^2 \xrightarrow[T\to\infty]{} \beta 
\quad \text{and} \quad 
\frac 1T \tr (\diag(D_T))^2 \xrightarrow[T\to\infty]{} \phi. 
\end{equation*}
Consider the matrices $S_T$ defined by \eqref{model}. Then
\begin{equation*}
\sqrt{T} \Bigl(S_T^{\sf H} D_T S_T - \frac{\tr D_T}{T} I_K \Bigr) 
\tolawlong G 
\end{equation*}
where $G = [G_{ij}]_{1\leq i,j\leq K}$ is random Hermitian such that $\{ G_{ij} \}_{i\leq j}$ are independent, $G_{ii} \sim {\cal N}(0, \beta + \kappa\phi)$ for $1\leq i\leq K$, and $G_{ij} \sim {\cal CN}(0, \beta)$ for $1\leq i < j \leq K$.
\end{lemma}

\begin{lemma}
\label{lm-diag}
Let $1 \leq k \leq s$ and
\begin{equation*}
D_{T} = 
p_k m_T(\rho_{k,T}) ( I_T + c_T m_T(\rho_{k,T}) R_T)^{-1}. 
\end{equation*}
Then $\limsup_T \| D_T \| < \infty$,  
\begin{equation*}
\frac 1T \tr (D_T^2) \tolong \beta_k \ \text{, and} \ 
\frac 1T \tr (\diag(D_T))^2 \tolong 
\phi_k
\end{equation*}
where $\beta_k$ and $\phi_k$ are given in Theorem~\ref{th-fluct-spikes}. 
\end{lemma} 

\begin{lemma}
\label{lm:fq-totale}
% Consider Model \eqref{model}. Write $A_T = U_T B_T^{\sf H}$ where 
% $U_T = H_T (H_T^{\sf H} H_T)^{-1/2}$ and 
% $B_T = S_T P_T^{1/2} (H_T^{\sf H} H_T)^{1/2}$. Partition $B_T$ as 
% $B_T = [ B_{1,T} \, \cdots \, B_{t,T} ]$ where the block $B_{k,T}$ is a 
% $T \times j_k$ matrix. 
Let $M_1, \ldots, M_t$, $M_k = [ M_{\ell,m,k} ]_{1\leq\ell,m\leq j_k}$, be random independent Hermitian matrices such that the $\{ M_{\ell,m,k} \}_{\ell\leq m}$ are independent, $M_{\ell,\ell,k} \sim {\mathcal N}(0, \beta_k + \kappa \phi_k)$, and $M_{\ell,m,k} \sim \mathcal{CN}(0, \beta_k)$ for $1\leq \ell < m \leq j_k$. Then 
\begin{equation*}
(\sqrt{T} F_{k,T})_{k=1,\ldots, t}\tolawlong (M_k)_{k=1,\ldots, t}.
\end{equation*}
\end{lemma} 

\revju{Theorem 2.3 of \cite{ChapCouHac'12} can be adapted to obtain the 
following result:}\footnote{
In fact, \cite[Th. 2.3]{ChapCouHac'12} characterizes the asymptotic fluctuations of the random variables $\sqrt{T}(\hat\lambda_{i,T} - {\brho}_{{\cal K}(i)})$ instead of the $\sqrt{T}(\hat\lambda_{i,T} - \rho_{{\cal K}(i),T})$, so that the speed of convergence of $\nu_T$ towards $\bnu$ and of $c_T$ towards $\bc$ had to be controlled through \cite[Assumption 7]{ChapCouHac'12}. By replacing ${\brho}_k$ with $\rho_{k,T}$, the proof of \cite[Th. 2.3]{ChapCouHac'12} goes on without the need for that assumption. Replacing ${\brho}_k$ by $\rho_{k,T}$ is enough for the present purpose. 
}
\begin{proposition}
\label{chapon} 
In the setting of Theorem~\ref{th-fluct-spikes}, let $G_1,\ldots, G_s$, $G_k\in\CC^{j_k\times j_k}$, be independent GUE matrices. Then, for any bounded and continuous $f:\R^{j_1+\cdots+j_s} \to \R$,  
\begin{equation*}
\E[ f(\eta_{1,T}, \ldots, \eta_{s,T}) ] -  \E[ f(\zeta_1, \ldots, \zeta_s) ] \to 0
\end{equation*}
where $\zeta_k$ is the random vector of the decreasingly ordered eigenvalues of $(p_k \bg(\brho_k)')^{-1} ( \sqrt{\alpha_k} G_k + \sqrt{T} F_{k,T} )$. 
\end{proposition} 

\revju{By Lemma \ref{lm:fq-totale}, the $s$-uple of matrices 
$(\sqrt{\alpha_k} G_k + \sqrt{T} F_{k,T})_{k=1}^s$ converges in distribution 
to the $s$-uple $(M_1,\ldots, M_s)$ provided in the statement of Theorem 
\ref{th-fluct-spikes}. Applying Proposition \ref{chapon}, this theorem is 
proven.}  

\section{Simulation results}\label{simulation}
We consider the setting of Section~\ref{example-doa}, with signals $s_{t,k}$ drawn from a QPSK constellation for which $\kappa=-1$. The signal power $a_k^2$ defines the signal-to-noise ratio (SNR). The noise is issued from an autoregressive (AR) process of order $1$ and parameter $a$, so that $[R_T]_{k,l}=a^{|k-l|}$. All other parameters are given in the figure captions.

\revju{In Figure~\ref{det3}, the probability of correct order estimation of the estimator proposed in Proposition~\ref{1stord-M} is compared against the MDL and AIC criteria, for $K=2$ equal power sources, for growing $N$, and for $c_T=0.5$ fixed. We observe that the proposed estimator outperforms the MDL and the AIC methods, consistently with the known inappropriateness of the latter. Note that the AIC particularly fails to detect any source, irrespective of $N$. 

In Figure~\ref{det1}, the false alarm rate (FAR) and correct detection rate (CDR) for single source detection is evaluated for different values of $\varepsilon$ and for growing ratios $c_T$. We observe here the impact of an appropriate choice of $\varepsilon$ which, if too small, generates a high FAR when the noise eigenvalues tend to spread (\emph{i.e.} for $c_T$ large) while, if too large, does not allow for correct source detection close to the detectability threshold (\emph{i.e.} for $c_T$ large).}
%In Figure~\ref{det1}, the receiver cooperation characteristics (ROC), parameterized by $\varepsilon$, for different values of $a$ are depicted. We compare here our proposed detection scheme against an oracle method which assumes perfect knowledge of $R_T$ that is used to whiten $Y_T$ before applying the proposed schemes in the white noise case. We observe that the proposed detector deteriorates with growing $a$, which can be explained by the natural spread of $\support(\mu)$ with $a$ large, implying larger inter-eigenvalue spacings within the noise subspace and therefore reduced efficiency of the detection test. On the opposite, the oracle estimator benefits from increased values of $a$, due to the SNR gain obtained by the whitening procedure. Observe that both approaches perform identically for $a=0$, which is expected since the system models in both cases are identical.

Figure~\ref{power} depicts the normalized mean square error (NMSE) $\E [(\hat{a}^2_1-a_1^2)^2a_1^{-4}]$ of the power estimation of Proposition~\ref{1stord-M} against its theoretical value obtained from Theorem~\ref{th-fluct-spikes}. For the purpose of analysis, we assume that the source is always detected, {\it i.e.} $\hat{k}_T=1$, irrespective of the SNR. As confirmed by Proposition~\ref{limit_snr}, the theoretical variance diverges as $p_k \downarrow p_{\rm lim}$. We however observe that in the finite $N,T$ regime, the power estimator errors remain bounded at low SNR. This is explained by the fact that, while the theoretical error diverges due to $\bDelta\downarrow 0$ (see Lemma~\ref{lmDelta}) as $p_k \downarrow p_{\rm lim}$, its estimator for each $N,T$ (obtained by replacing $\bm$ by $\hat{m}_T$) is always non-zero even for $p_k=p_{\rm lim}$. In the high SNR regime, here with $\kappa=-1$, the NMSE becomes linear (in dB scale) with slope $-10$ dB/decade. It is easily shown that the limiting SNR gap between the proposed and oracle estimators is exactly 
\begin{equation*}
	10\log \Big( \int_0^1|\bp(\exp(2\imath \pi u))|^2du \cdot \int_0^1 |\bp(\exp(2\imath\pi u))|^{-2}du\Big)~\text{dB}
\end{equation*}
which is merely due to a gain in SNR after whitening. In particular, the larger the correlation parameter $a$, the bigger the limiting gap.

In Figure~\ref{doa2}, the mean square error $\E[(\hat{\gamma}(\theta_1)-{\gamma}(\theta_1))^2]$ of the localization function at position $\theta_1=10^\circ$ is compared against the performances of the oracle estimator (which performs pre-whitening prior to using the estimator of \cite{HachLouMes'11} or equivalently that of Proposition~\ref{prop:GMUSIC}) and of the traditional MUSIC estimator with localization function $\hat{\gamma}_{ {\rm trad},T}(\theta)\triangleq \sum_{k=1}^{\hat{k}_T} h(\theta)^{\sf H} \hat{u}_{k,T}\hat{u}_{k,T}^{\sf H}h(\theta)$ in the notations of Proposition~\ref{1stord-M}. The source is again supposed always detected so that $\hat{k}_T=1$ throughout the experiment. The proposed estimator outperforms greatly the traditional MUSIC approach here, which is both due to the large $N,T$ regime improvement and to the consideration of the non-white noise setting. The oracle estimator shows a huge performance improvement in the low SNR regime, which translates the fact that condition \eqref{cond:spike} (which needs to be fulfilled for either method to be valid) is extremely demanding when $a=0.6$ (due to $\support(\bmu)$ being large). In the large SNR regime, a constant gap is maintained which, although we do not provide theoretical support, appears as a similar SNR-gap phenomenon as observed in Figure~\ref{power}.

In Figure~\ref{doa1}, we now take $K=2$ sources, with $a_1=a_2$ the amplitude of which define the SNR, and again assuming $\hat{k}_T=2$. Here are compared the performances of resolution of two close sources located at $\theta_1=10^{\circ}$ and $\theta_2=12^{\circ}$ for the localization method proposed in Proposition~\ref{prop:GMUSIC}, for the oracle estimator, and for the traditional MUSIC estimator. The figure of merit, referred to as resolution probability, is the probability of identifying exactly two local minima of the localization function in the window $[5^{\circ},17^{\circ}]$. We observe that the proposed algorithm performs significantly better than the traditional MUSIC method, confirming the results of \cite{HachLouMes'11} for the current model. 

%%%%%%%%%%%%%%%%COMPARISON DETECTORS %%%%%%%%%%%%%%%%%%%%%%%%%%%%%%%%%%%%%%%
\revju{
\begin{figure}[H]
\center
  \begin{tikzpicture}[font=\footnotesize]
    \renewcommand{\axisdefaulttryminticks}{2} 
    \pgfplotsset{every axis/.append style={mark options=solid, mark size=2pt}}
    \tikzstyle{every major grid}+=[style=densely dashed]       
%    \tikzstyle{every pin}=[fill=white,draw=black,font=\footnotesize,edge style={<-}] 
\pgfplotsset{every axis legend/.append style={fill=white,cells={anchor=west},at={(0.99,0.05)},anchor=south east}}    
    \tikzstyle{every axis y label}+=[yshift=-10pt] 
    \tikzstyle{every axis x label}+=[yshift=5pt]
    \begin{axis}[
      grid=major,
      %ymajorgrids=false,
      xlabel={$N$},
      ylabel={Probability of correct order estimation},
      xmin=12,
      xmax=30, 
      ymin=0, 
      ymax=1,
      width=0.7\columnwidth,
      height=0.45\columnwidth
      ]
%      \addplot[smooth,black,line width=0.5pt] plot coordinates{
%(12.000000,0.817100)(13.000000,0.884100)(14.000000,0.933100)(15.000000,0.969400)(16.000000,0.984600)(17.000000,0.992600)(18.000000,0.996100)(19.000000,0.998000)(20.000000,0.999000)(21.000000,0.999400)(22.000000,0.999300)(23.000000,0.999600)(24.000000,0.999800)(25.000000,1.000000)(26.000000,0.999900)(27.000000,1.000000)(28.000000,0.999800)(29.000000,1.000000)(30.000000,0.999500)
%
%
%
%      };
%      \addplot[smooth,black,line width=0.5pt,mark=square] plot coordinates{
%(12.000000,0.510600)(13.000000,0.534700)(14.000000,0.555500)(15.000000,0.569300)(16.000000,0.586600)(17.000000,0.606400)(18.000000,0.620900)(19.000000,0.643900)(20.000000,0.665300)(21.000000,0.684200)(22.000000,0.705900)(23.000000,0.723700)(24.000000,0.743300)(25.000000,0.755800)(26.000000,0.770800)(27.000000,0.783900)(28.000000,0.796500)(29.000000,0.812600)(30.000000,0.834200)
%
%
%
%      };
%    \addplot[smooth,black,line width=0.5pt,mark=o] plot coordinates{
%(12.000000,0.021200)(13.000000,0.012500)(14.000000,0.007800)(15.000000,0.003900)(16.000000,0.002200)(17.000000,0.001500)(18.000000,0.000500)(19.000000,0.000300)(20.000000,0.000100)(21.000000,0.000100)(22.000000,0.000000)(23.000000,0.000000)(24.000000,0.000000)(25.000000,0.000000)(26.000000,0.000000)(27.000000,0.000000)(28.000000,0.000000)(29.000000,0.000000)(30.000000,0.000000)
%
%
%
%      };
      \addplot[smooth,black,line width=0.5pt] plot coordinates{
	      (12.000000,0.586200)(13.000000,0.694000)(14.000000,0.790500)(15.000000,0.858500)(16.000000,0.910500)(17.000000,0.948300)(18.000000,0.970100)(19.000000,0.980000)(20.000000,0.990200)(21.000000,0.992100)(22.000000,0.994700)(23.000000,0.997500)(24.000000,0.997400)(25.000000,0.998300)(26.000000,0.999300)(27.000000,0.999000)(28.000000,0.999100)(29.000000,0.998000)(30.000000,0.998100)

      };
      \addplot[smooth,black,line width=0.5pt,mark=square] plot coordinates{
	      (12.000000,0.510600)(13.000000,0.534700)(14.000000,0.555500)(15.000000,0.569300)(16.000000,0.586600)(17.000000,0.606400)(18.000000,0.620900)(19.000000,0.643900)(20.000000,0.665300)(21.000000,0.684200)(22.000000,0.705900)(23.000000,0.723700)(24.000000,0.743300)(25.000000,0.755800)(26.000000,0.770800)(27.000000,0.783900)(28.000000,0.796500)(29.000000,0.812600)(30.000000,0.834200)

      };
      \addplot[smooth,black,line width=0.5pt,mark=o] plot coordinates{
	      (12.000000,0.020500)(13.000000,0.013800)(14.000000,0.008900)(15.000000,0.002900)(16.000000,0.002000)(17.000000,0.001000)(18.000000,0.001000)(19.000000,0.000400)(20.000000,0.000200)(21.000000,0.000200)(22.000000,0.000000)(23.000000,0.000000)(24.000000,0.000000)(25.000000,0.000100)(26.000000,0.000000)(27.000000,0.000000)(28.000000,0.000000)(29.000000,0.000000)(30.000000,0.000000)

      };

	\legend{{Proposed},{MDL},{AIC}}
       \end{axis}
       \end{tikzpicture}
       \caption{Probability of correct order estimation versus $N$ with $K=2$, SNR$=10$ dB (same power for each source), $L=5$, $\varepsilon=0.75$, $c_T=0.5$, and $a=0.6$.}
\label{det3}
\end{figure}
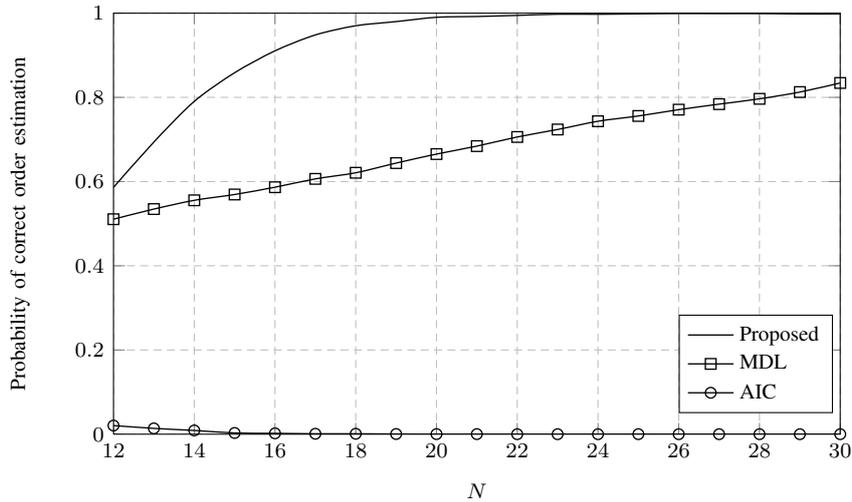
%%%%%%%%%%%%%%%%%%%%%%%%%%%% DETECTION ROC %%%%%%%%%%%%%%%%%%%%%%%%%%%%%%%%%%%%%%%%
\begin{figure}[H]
\center
	\begin{tikzpicture}[font=\footnotesize]
		\renewcommand{\axisdefaulttryminticks}{2} 
		\pgfplotsset{every axis/.append style={mark options=solid, mark size=2pt}}
		\tikzstyle{every major grid}+=[style=densely dashed]       
			      %    \tikzstyle{every pin}=[fill=white,draw=black,font=\footnotesize,edge style={<-}] 
		\pgfplotsset{every axis legend/.append style={fill=white,cells={anchor=west},at={(0.01,0.2)},anchor=south west}}    
		\tikzstyle{every axis y label}+=[yshift=-10pt] 
		\tikzstyle{every axis x label}+=[yshift=5pt]
		\begin{axis}[
				grid=major,
								      %ymajorgrids=false,
				xlabel={$c_T$},
				ylabel={Probability},
				xtick = {1,0.9,0.8,0.7,0.6,0.5,0.4,0.3,0.2,0.1},
				xmin=0.1,
				xmax=1, 
				ymin=0, 
				ymax=1,
				width=0.7\columnwidth,
				height=0.45\columnwidth
			]
			\addplot[smooth,black,line width=0.5pt,mark=o] plot coordinates{
				(0.100000,1.000000)(0.200000,1.000000)(0.300000,1.000000)(0.400000,1.000000)(0.500000,1.000000)(0.600000,1.000000)(0.700000,1.000000)(0.800000,1.000000)(0.900000,1.000000)(1.000000,1.000000)

			};
			\addplot[smooth,black,densely dashed,line width=0.5pt,mark=o] plot coordinates{
				(0.100000,0.000000)(0.200000,0.001600)(0.300000,0.015000)(0.400000,0.060200)(0.500000,0.156300)(0.600000,0.292400)(0.700000,0.434800)(0.800000,0.594400)(0.900000,0.725700)(1.000000,0.828800)

			};
			\addplot[smooth,black,line width=0.5pt,mark=square] plot coordinates{
				(0.100000,1.000000)(0.200000,1.000000)(0.300000,1.000000)(0.400000,1.000000)(0.500000,0.999800)(0.600000,0.998800)(0.700000,0.993800)(0.800000,0.978700)(0.900000,0.954600)(1.000000,0.927100)

			};
			\addplot[smooth,black,densely dashed,line width=0.5pt,mark=square] plot coordinates{
				(0.100000,0.000000)(0.200000,0.000000)(0.300000,0.000000)(0.400000,0.000000)(0.500000,0.000200)(0.600000,0.002000)(0.700000,0.003000)(0.800000,0.009400)(0.900000,0.020000)(1.000000,0.041400)

			};
			\addplot[smooth,black,line width=0.5pt,mark=star] plot coordinates{ 
				(0.100000,1.000000)(0.200000,1.000000)(0.300000,0.995600)(0.400000,0.934700)(0.500000,0.761400)(0.600000,0.549800)(0.700000,0.393600)(0.800000,0.288500)(0.900000,0.217500)(1.000000,0.164800)

			};
			\addplot[smooth,black,densely dashed,line width=0.5pt,mark=star] plot coordinates{
				(0.100000,0.000000)(0.200000,0.000000)(0.300000,0.000000)(0.400000,0.000000)(0.500000,0.000000)(0.600000,0.000000)(0.700000,0.000000)(0.800000,0.000000)(0.900000,0.000000)(1.000000,0.000000)

			};
			\legend{ {CDR $\varepsilon=0.5$},{FAR $\varepsilon=0.5$},{CDR $\varepsilon=1$},{FAR $\varepsilon=1$},{CDR $\varepsilon=2$},{FAR $\varepsilon=2$}}

		\end{axis}
	\end{tikzpicture}

\caption{CDR (plain curve) and FAR (dashed curves) versus $c_T$ with $K=1$, $N=20$, SNR$=10$ dB, $L=5$, and $a=0.6$.}
	\label{det1}
\end{figure}
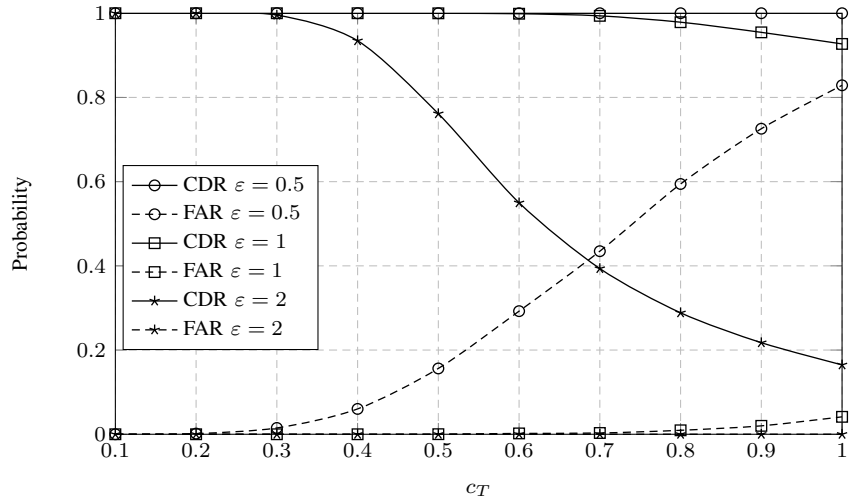
}
%%%%%%%%%%%%%%%%%%%%%%%%% POWER ESTIMATION %%%%%%%%%%%%%%%%%%%%%%%%%%%%%%%%%%%%%%%%%%%%%%%%%%%%%%
\begin{figure}[H]
\center
  \begin{tikzpicture}[font=\footnotesize]
    \renewcommand{\axisdefaulttryminticks}{2} 
    \tikzstyle{every major grid}+=[style=densely dashed]       
%    \tikzstyle{every pin}=[fill=white,draw=black,font=\footnotesize,edge style={<-}] 
\pgfplotsset{every axis legend/.append style={fill=white,cells={anchor=west},at={(0.99,0.99)},anchor=north east}}    
    \tikzstyle{every axis y label}+=[yshift=-10pt] 
    \tikzstyle{every axis x label}+=[yshift=5pt]
    \begin{semilogyaxis}[
      grid=major,
      %ymajorgrids=false,
      xlabel={SNR (dB)},
      ylabel={NMSE},
      xmin=0,
      xmax=12, 
      ymin=0.001, 
      width=0.7\columnwidth,
      height=0.45\columnwidth,
      ymax=1
      ]
      \addplot[smooth,black,line width=0.5pt] plot coordinates{
(0.000000,1.272169)(0.500000,0.968526)(1.000000,0.625557)(1.500000,0.406299)(2.000000,0.256727)(2.500000,0.178447)(3.000000,0.120174)(3.500000,0.090211)(4.000000,0.070609)(4.500000,0.055674)(5.000000,0.043802)(5.500000,0.035549)(6.000000,0.027363)(6.500000,0.022588)(7.000000,0.018147)(7.500000,0.015267)(8.000000,0.012441)(8.500000,0.010367)(9.000000,0.008823)(9.500000,0.007500)(10.000000,0.006363)(10.500000,0.005594)(11.000000,0.004871)(11.500000,0.004352)(12.000000,0.003809)
      };      
      \addplot[smooth,black,line width=0.5pt,mark=x] plot coordinates{    
(0.000000,1.258506)(0.500000,1.333790)(1.000000,1.420759)(1.500000,1.521504)(2.000000,1.638506)(2.500000,1.774796)(3.000000,1.934005)(3.500000,0.246901)(4.000000,0.117289)(4.500000,0.072280)(5.000000,0.049930)(5.500000,0.036805)(6.000000,0.028294)(6.500000,0.022398)(7.000000,0.018117)(7.500000,0.014897)(8.000000,0.012409)(8.500000,0.010442)(9.000000,0.008862)(9.500000,0.007572)(10.000000,0.006508)(10.500000,0.005620)(11.000000,0.004873)(11.500000,0.004240)(12.000000,0.003701)
      };
      \addplot[smooth,black,line width=0.5pt,mark=o] plot coordinates{
(0.000000,0.035413)(0.500000,0.030118)(1.000000,0.025651)(1.500000,0.023826)(2.000000,0.020180)(2.500000,0.017136)(3.000000,0.014248)(3.500000,0.013040)(4.000000,0.011270)(4.500000,0.010155)(5.000000,0.008481)(5.500000,0.007383)(6.000000,0.006701)(6.500000,0.005951)(7.000000,0.005257)(7.500000,0.004674)(8.000000,0.004070)(8.500000,0.003528)(9.000000,0.003346)(9.500000,0.002855)(10.000000,0.002502)(10.500000,0.002276)(11.000000,0.001928)(11.500000,0.001746)(12.000000,0.001604)
      }; 
      \addplot[smooth,black,line width=0.5pt,mark=star] plot coordinates{
(0.000000,0.037101)(0.500000,0.031319)(1.000000,0.026855)(1.500000,0.022761)(2.000000,0.019586)(2.500000,0.017049)(3.000000,0.014654)(3.500000,0.012912)(4.000000,0.011204)(4.500000,0.009822)(5.000000,0.008625)(5.500000,0.007535)(6.000000,0.006615)(6.500000,0.005856)(7.000000,0.005187)(7.500000,0.004564)(8.000000,0.004057)(8.500000,0.003584)(9.000000,0.003167)(9.500000,0.002816)(10.000000,0.002514)(10.500000,0.002217)(11.000000,0.001966)(11.500000,0.001743)(12.000000,0.001561)
      };
     
      \legend{ {Proposed},{Proposed (theory)},%{Proposed estimator a=0.8},
      {Oracle},{Oracle (theory)}}%,{Whitened estimator a=0.8}}
    \end{semilogyaxis}
  \end{tikzpicture}
  \vspace{-0.3cm}
  \caption{NMSE of the estimated power versus SNR with $K=1$, $N=20$, $c_T=0.5$, and $a=0.6$.}
  \label{power}
\end{figure}
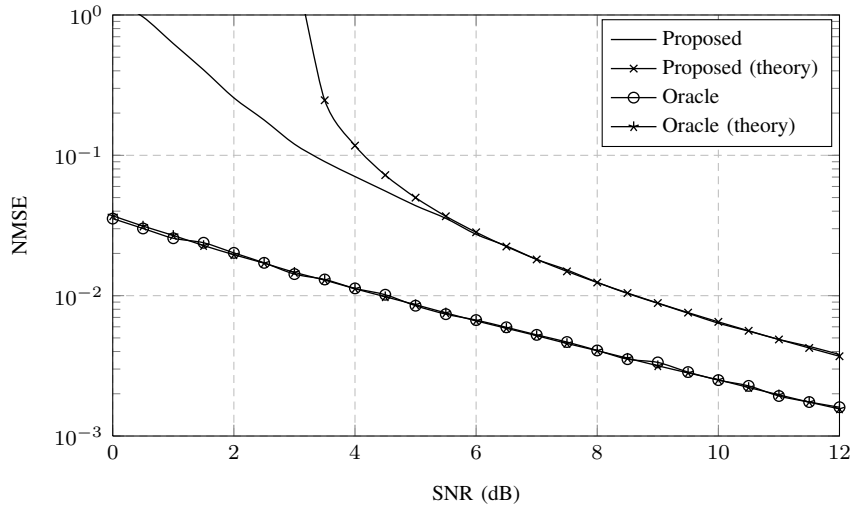
%%%%%%%%%%%%%%%%%%%%%%%%%% LOCALIZATION FUNCTION %%%%%%%%%%%%%%%%%%%%%%%%%%%%%%%%%%
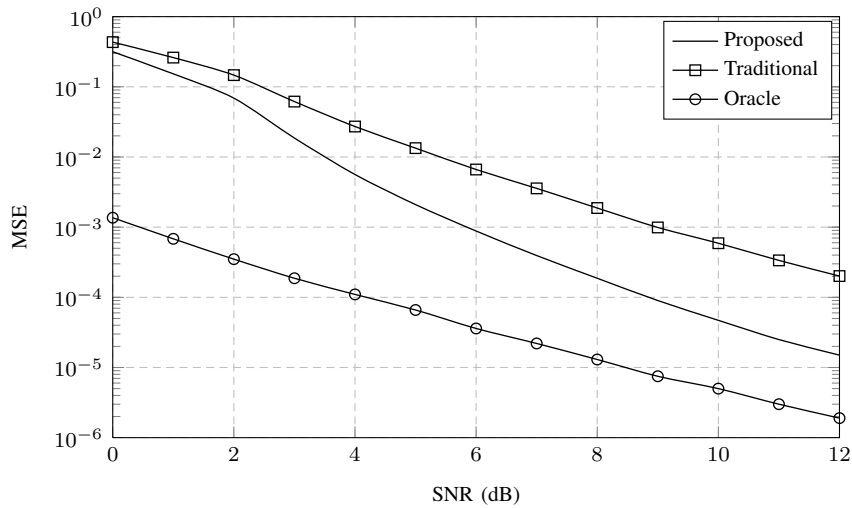
\begin{figure}[H]
\center
  \begin{tikzpicture}[font=\footnotesize]
    \renewcommand{\axisdefaulttryminticks}{2} 
    \tikzstyle{every major grid}+=[style=densely dashed]       
%    \tikzstyle{every pin}=[fill=white,draw=black,font=\footnotesize,edge style={<-}] 
\pgfplotsset{every axis legend/.append style={fill=white,cells={anchor=west},at={(0.99,0.99)},anchor=north east}}    
    \tikzstyle{every axis y label}+=[yshift=-10pt] 
    \tikzstyle{every axis x label}+=[yshift=5pt]
    \begin{semilogyaxis}[
      grid=major,
      %ymajorgrids=false,
      xlabel={SNR (dB)},
      ylabel={MSE},
      xmin=0,
      xmax=12, 
      ymax=1,
      ymin=1e-6,
      width=0.7\columnwidth,
      height=0.45\columnwidth, 
      ymax=1
      ]
      \addplot[smooth,black,line width=0.5pt] plot coordinates{
(0.000000,0.315907)(1.000000,0.153017)(2.000000,0.068990)(3.000000,0.018610)(4.000000,0.005627)(5.000000,0.002096)(6.000000,0.000881)(7.000000,0.000394)(8.000000,0.000187)(9.000000,0.000090)(10.000000,0.000047)(11.000000,0.000025)(12.000000,0.000015)(13.000000,0.000009)(14.000000,0.000005)(15.000000,0.000003)(16.000000,0.000002)(17.000000,0.000001)(18.000000,0.000001)(19.000000,0.000000)(20.000000,0.000000)(21.000000,0.000000)(22.000000,0.000000)(23.000000,0.000000)(24.000000,0.000000)(25.000000,0.000000)
      };
      \addplot[smooth,black,line width=0.5pt,mark=square] plot coordinates{
(0.000000,0.433365)(1.000000,0.260837)(2.000000,0.146764)(3.000000,0.061611)(4.000000,0.027175)(5.000000,0.013368)(6.000000,0.006628)(7.000000,0.003570)(8.000000,0.001873)(9.000000,0.000990)(10.000000,0.000590)(11.000000,0.000336)(12.000000,0.000201)(13.000000,0.000122)(14.000000,0.000073)(15.000000,0.000044)(16.000000,0.000026)(17.000000,0.000017)(18.000000,0.000010)(19.000000,0.000006)(20.000000,0.000004)(21.000000,0.000003)(22.000000,0.000002)(23.000000,0.000001)(24.000000,0.000001)(25.000000,0.000000)
      };
      \addplot[smooth,black,line width=0.5pt,mark=o] plot coordinates{
(0.000000,0.001363)(1.000000,0.000682)(2.000000,0.000350)(3.000000,0.000187)(4.000000,0.000110)(5.000000,0.000066)(6.000000,0.000036)(7.000000,0.000022)(8.000000,0.000013)(9.000000,0.0000075)(10.000000,0.000005)(11.000000,0.000003)(12.000000,0.0000019)(13.000000,0.0000012)(14.000000,0.0000008)(15.000000,0.000000)(16.000000,0.000000)(17.000000,0.000000)(18.000000,0.000000)(19.000000,0.000000)(20.000000,0.000000)(21.000000,0.000000)(22.000000,0.000000)(23.000000,0.000000)(24.000000,0.000000)(25.000000,0.000000)
      };

      \legend{ {Proposed},{Traditional},{Oracle}}
    \end{semilogyaxis}
  \end{tikzpicture}
  \vspace{-0.3cm}
  \caption{MSE of the localization function versus SNR with $K=1$, $N=20$, $c_T=0.2$, and $a=0.6$.}
  \label{doa2}
\end{figure}
%%%%%%%%%%%%%%%%%%%%%%%%%%%%% RESOLUTION PROBA %%%%%%%%%%%%%%%%%%%%%%%%%%%%%%%%%%%%%%%%%%%%%%%%%%%%
\begin{figure}[H]
\center
  \begin{tikzpicture}[font=\footnotesize]
    \renewcommand{\axisdefaulttryminticks}{2} 
    \tikzstyle{every major grid}+=[style=densely dashed]       
%    \tikzstyle{every pin}=[fill=white,draw=black,font=\footnotesize,edge style={<-}] 
\pgfplotsset{every axis legend/.append style={fill=white,cells={anchor=west},at={(0.99,0.02)},anchor=south east}}    
    \tikzstyle{every axis y label}+=[yshift=-10pt] 
    \tikzstyle{every axis x label}+=[yshift=5pt]
     \begin{axis}[
      grid=major,
      %ymajorgrids=false,
      xlabel={SNR (dB)},
      ylabel={Resolution probability},
      xmin=10,
      xmax=22,
      ymax=1, 
      ymin=0,
      width=0.7\columnwidth,
      height=0.45\columnwidth, 
      ymax=1
      ]
      \addplot[smooth,black,line width=0.5pt] plot coordinates{
(10.000000,0.436500)(11.000000,0.521900)(12.000000,0.615400)(13.000000,0.711100)(14.000000,0.804400)(15.000000,0.886400)(16.000000,0.951800)(17.000000,0.984200)(18.000000,0.996100)(19.000000,0.999500)(20.000000,1.000000)(21.000000,1.000000)(22.000000,1.000000)(23.000000,1.000000)(24.000000,1.000000)(25.000000,1.000000)
      };
      \addplot[smooth,black,line width=0.5pt,mark=square] plot coordinates{
(10.000000,0.000300)(11.000000,0.000200)(12.000000,0.001200)(13.000000,0.010100)(14.000000,0.051500)(15.000000,0.166100)(16.000000,0.386300)(17.000000,0.665300)(18.000000,0.873000)(19.000000,0.967600)(20.000000,0.997200)(21.000000,0.999900)(22.000000,1.000000)(23.000000,1.000000)(24.000000,1.000000)(25.000000,1.000000)
      };
      \addplot[smooth,black,line width=0.5pt,mark=o] plot coordinates{
(10.000000,0.852600)(11.000000,0.926200)(12.000000,0.967800)(13.000000,0.986100)(14.000000,0.996300)(15.000000,0.999700)(16.000000,1.000000)(17.000000,1.000000)(18.000000,1.000000)(19.000000,1.000000)(20.000000,1.000000)(21.000000,1.000000)(22.000000,1.000000)(23.000000,1.000000)(24.000000,1.000000)(25.000000,1.000000)
      };

      \legend{ {Proposed},{Traditional},{Oracle}}
    \end{axis}
  \end{tikzpicture}
  \vspace{-0.3cm}
  \caption{Resolution probability versus SNR with $K=2$, $N=20$, $c_T=0.2$, and $a=0.6$.}
  \label{doa1}
\end{figure}
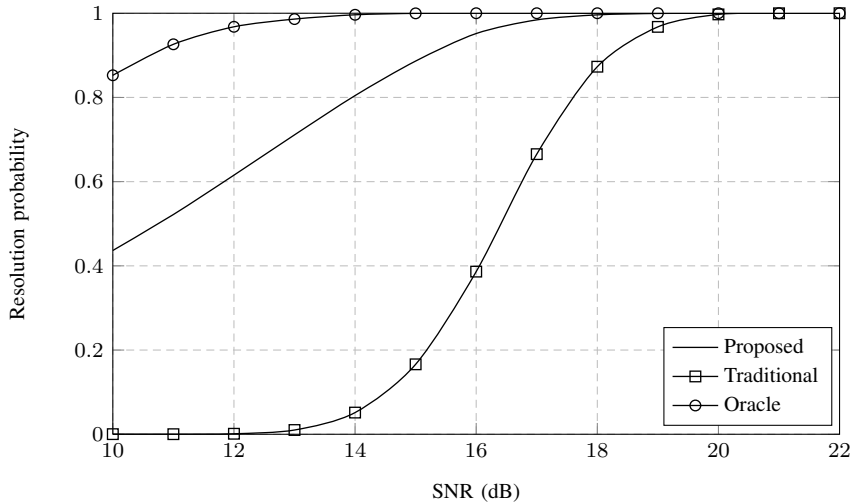

\section{Conclusion and research prospects} 
\label{sec:conclusion}
This article introduced a novel set of statistical inference methods for large dimensional information-plus-noise models with multiple sources, \revju{when the noise is correlated in time while the information is correlated in space (or vice-versa)}. These techniques were proved consistent in the limiting regime where both the system size and the number of observations go large. The approach pursued here relies on the asymptotic spectral separation between noise and signal in the observed sample covariance matrix. Under the same hypotheses, using instead prior information on the noise structure, an alternative approach could consist in estimating the noise covariance in the presence of signals, similar to \cite{bickel2008regularized} which treats the noise-only case. It is expected that this approach performs better in the low SNR regime, resurrecting signals unseen by our current method. In the high SNR regime, the covariance estimation will instead be too degraded for this method to be beneficial. A trade-off is therefore expected between both approaches, which we shall study in a future work.

In the specific problem of signal detection, the choice of the eigenvalue ``gap parameter'' $\varepsilon$ does not account for the observation of the small eigenvalues of $Y_TY_T^\herm$ as for the power and direction-of-arrival estimation techniques (through $\hat{m}_T$). It seems nonetheless natural to be able to evaluate the right-edge of $\support(\bmu)$ from these eigenvalues, thus resulting in a test to compare $\hat{\lambda}_{i,T}$, $i=1,\ldots,L$, to the estimated edge. To finely tune the test, one can then use the results from \cite{elkar'07} which proves Tracy-Widom fluctuations at the edge with scaling coefficient $\underline x''(m_b)$ ($m_b$ given by Corollary~\ref{cor:edge}). However, estimating both the edge and this coefficient constitute a challenging problem so far.

\appendix

\subsection{Proof of Corollary \ref{cor:edge}} 
\label{prf:coredge} 
The derivative 
\begin{equation*}
\bx'(m) = \frac{1}{m^2} 
- \bc \int \left( \frac{t}{1+\bc m t} \right)^2 \bnu(dt) 
\end{equation*}
of $\bx(m)$ is continuous and increasing on $( -(\bc b_\bnu)^{-1}, 0)$, and $\bx'(m)\to\infty$ as $m \uparrow 0$. To establish the proposition, it will be enough to show that \revju{$\bx'(m) \to -\infty$} as $m \downarrow -(\bc b_\bnu)^{-1}$. This is obvious when $\bnu(b_\bnu) > 0$. Assume then $\bnu(b_\bnu) = 0$. When $m \downarrow -(\bc b_\bnu)^{-1}$, by the monotone convergence theorem 
\begin{align*} 
\int \frac{t^2}{(1+\bc m t)^2} \bnu(dt) 
\uparrow 
\int \frac{t^2}{(1- t/b)^2} \bnu(dt) \geq \int_{[ b_\bnu-\varepsilon, b_\bnu]} 
\frac{b^2t^2}{(b- t)^2} f_\bnu(t) \, dt = \infty
\end{align*} 
from the behavior of $f_\bnu(t)$ near $b_\bnu$, which proves the result. 

\subsection{Proof of Lemma \ref{lmDelta}}
\label{prf-Delta}

Considering Equation \eqref{m=f(m)}, we obtain after some calculus that $\bm'(x) = \bm^2(x) / {\bs\Delta}(x)$ on $(b,\infty)$. Since $\bm(x)$ is negative and increasing on $(b,\infty)$, both $\bm'(x)$ and $\bm^2(x)$ are positive on this interval so that ${\bs\Delta}(x) > 0$ on $(b,\infty)$. \\
Proposition~\ref{prop:edge} shows that $b$ coincides with the minimum of $\bx(m)$ on $((-\bc b_\bnu)^{-1}, 0)$. Moreover, when Assumption~\ref{nu:edge} is satisfied (which is the case for the model \eqref{model} by Lemma~\ref{lm:model}), the proof of Corollary~\ref{cor:edge} shows that $\bx(m)$ attains its minimum at a unique point $m_b \in ((-\bc b_\bnu)^{-1}, 0)$, and $\bx'(m_b) = 0$. Finally, Proposition~\ref{prop:edge} shows that $\bx(m)$ is the inverse of $\bm(x)$ on $(b,\infty)$. It results that $\bm(x)\to m_b$ and $\bm'(x) = 1 / \bx'(\bm(x))\to\infty$ as $x\downarrow b$. This proves ${\bs\Delta}(x) \to 0$ as $x\downarrow b$. \\
When $x\to\infty$, both $(x\bm(x))^2 = (\int x(t-x)^{-1} \bmu(dt))^2$ and $x^2\bm'(x) = \int x^2 (t-x)^{-2} \bmu(dt)$ converge to $1$. Hence, ${\bs\Delta}(x) = (x\bm(x))^2 (x^2\bm'(x))^{-1} \to 1$, concluding the proof. 

\subsection{Theorem \ref{2ord-powers}: main steps of the proof} 
\label{prf:2ord-powers}
For simplicity, we focus on the fluctuations of $\sqrt{T}(\hat p_{1,T} - p_1)$. Recall that $\hat p_{1,T} = \hat g_T(\hat\lambda_{1,T})^{-1}$ and $p_1 = g_T(\rho_{1,T})^{-1}$. Define $\underline g_T(x) = \underline m_T(x) ( x c_T \underline m_T(x) + c_T - 1)$ with $\underline m_T(x)$ defined in Theorem~\ref{th:lsm}-\ref{cvg-unif}). We have 
\begin{align*} 
\sqrt{T}(\hat p_{1,T} - p_1) &= 
\sqrt{T}( \hat g_T(\hat\lambda_{1,T})^{-1} - g_T(\rho_{1,T})^{-1}) \\ 
&=  
\sqrt{T}( \hat g_T(\hat\lambda_{1,T})^{-1} 
                 - \underline g_T(\hat\lambda_{1,T})^{-1}) \\ 
&\phantom{=} + \sqrt{T}( \underline g_T(\hat\lambda_{1,T})^{-1} 
                        - g_T(\hat\lambda_{1,T})^{-1})  \\
&\phantom{=} + \sqrt{T}( g_T(\hat\lambda_{1,T})^{-1} - g_T(\rho_{1,T})^{-1}) \\
&\triangleq f_{1,T}(\hat\lambda_{1,T}) + f_{2,T}(\hat\lambda_{1,T}) + 
f_{3,T}(\hat\lambda_{1,T}) . 
\end{align*} 
As $\lambda_{1,T} \toasshort \brho_1$, we can replace $f_{1,T}(\hat\lambda_{1,T})$ by $f_{1,T}(\hat\lambda_{1,T}) \1_I(\lambda_{1,T})$ where $\1_I$ is the indicator function on a small compact interval $I$ in a neighborhood of $\brho_1$. Mimicking the proof of Theorem~\ref{th:estim-p}, we can show that $\sup_{x \in I} f_{1,T}(x) \toprobashort 0$. We similarly restrict $f_{2,T}$ to $I$. On this set, it is possible to show that the random process $T (\underline m_T(x) - m_T(x))$ valued in the set $C(I)$ of the continuous functions on $I$, converges in distribution towards a Gaussian process in $C(I)$. This result was shown in \cite{bai-sil-clt04} for $I$ a compact path of $\C_+$; this can be generalized to the interval $I$ of interest in this proof by using the Gaussian tools used in {\it e.g.} \cite{ChapCouHac'12}. As a result, $\sup_{x \in I} f_{2,T}(x) \toprobashort 0$. To deal with $f_{3,T}$, we start by observing that $g_T(\rho_{k,T}) \to \bg(\brho_k)$ and $(1/g_T(\rho_{k,T}))' \to - \bg'(\brho_k) / \bg^2(\brho_k) = - p_k^2 \bg'(\brho_k)$. Using the result of Theorem \ref{th-fluct-spikes} and applying the Delta method \cite[Prop.~6.1.6]{bro-dav-91}, we can show that $f_{3,T}(\hat\lambda_{1,T}) \tolawshort p_1 [M_1]_{11}$. The generalization to the vectors $\xi_{k,T}$ defined in the theorem shows no major difficulty. 

\subsection{Proof of Proposition \ref{limit_snr}}
\label{prf-limsnr}

From Theorem~\ref{th:1stord}, $\brho_k \downarrow b$ as $p_k \downarrow p_{\rm lim}$. Hence, by Lemma~\ref{lmDelta}, $\bDelta(\brho_k) \to 0$ as $p_k \downarrow p_{\rm lim}$. Moreover, the proof of this lemma shows that $|\bm(\rho_k)|$ remains bounded as $\brho_k \downarrow b$. Hence, since $\bnu \neq \delta_0$ by Assumption~\ref{ass:R}, the integrals in the expression of $\alpha_k$ are lower bounded by a positive number as $p_k \downarrow p_{\rm lim}$. Thus, $\alpha_k \to\infty$ which proves the first part of the lemma. \\
When $p_k \to\infty$, $\brho_k / p_k \to 1$ and $\brho_k \bm(\brho_k) \to -1$. Taking $p_k \to\infty$ into the expressions of the integrals on the right hand sides of the expressions of $\alpha_k$, $\beta_k$, and $\phi_k$ and recalling that $\bDelta(\brho_k) \to 1$, we get $\alpha_k \to 0$, $\beta_k \to 1$, and $\phi_k \to 1$, which proves the lemma. 

\subsection{Lemma \ref{tlc:fq}: sketch of the proof} 
\label{prf-tlc:fq} 

The fluctuations of quadratic forms of the type $s_T^H D_T s_T$ where $s_T \in \C^T$ has i.i.d. entries have been well studied ({\it e.g.} \cite[Th.~2.1]{BhanGirKok'07}, \cite[Th.~3]{KammKhaHac'09}). Here, the vector $s_T$ is replaced by the matrix $S_T\in\C^{T\times K}$ which introduces some differences in the proof. We follow here the lines of the proof of \cite[Th.~3]{KammKhaHac'09} and stress the main differences. 

Let $\sqrt{T} S_T^{\sf H} = [ {{\bf s}}_{1}, \cdots, {{\bf s}}_{T} ]$ where ${\bf s}_t = [ s_{t,1}^*,\ldots, s_{t,K}^* ]^{\sf T}$ and let $C = [c_{ij}]\in\C^{K\times K}$ Hermitian matrix. Showing that 
\begin{align*} 
\sqrt{T} \tr C \Bigl( S_T^{\sf H} D_T S_T - 
\frac1T{\tr D_T} I_K \Bigr)\tolawlong {\cal N}\left(0, \beta \tr(C^2) + \kappa\alpha \tr[(\diag(C))^2]
\right)  
\end{align*} 
and invoking the Cram\'er-Wold device establishes the lemma. \\ 
Consider the sequence of increasing $\sigma$-fields ${\cal F}_t = \sigma( {{\bf s}}_1, \ldots, {{\bf s}}_t )$, $t=1,\ldots, T$, and denote $\E_t$ the expectation conditional to ${\cal F}_t$. Then, with $\E_0=\E$, 
\begin{align*} 
\sqrt{T} \tr C \Bigl( S_T^{\sf H} D_T S_T - 
\frac1T\tr D_T I_K \Bigr)= \sqrt{T} \sum_{t=0}^{T-1} \left( \E_{t+1} - \E_{t} \right) 
\tr C S_T^{\sf H} D_T S_T
\end{align*} 
which is a sum of martingale increments, so that the key tool for establishing Lemma~\ref{tlc:fq} is martingale CLT \cite[Th.~35.12]{Bill'95}. Writing $Z_t = ( \E_{t+1} - \E_{t} ) \tr C S_T^{\sf H} D_T S_T$, we need to show:
\begin{itemize} 
\item {\it Lyapunov's condition} : there exists $\delta > 0$ for which 
\begin{equation*}
T^{1+\delta/2} \sum_{t=0}^{T-1} \E Z_t^{2+\delta} 
\tolong 0.
\end{equation*}
\item The following convergence holds
\begin{equation*}
T \sum_{t=0}^{T-1} \E_t Z_t^2 \toprobalong \beta \tr(C^2) + \kappa\alpha \tr[(\diag(C))^2] .
\end{equation*}
\end{itemize}
Taking $\delta = 2$ and mimicking the calculus of \cite[page 5058]{KammKhaHac'09} (based on Burkholder's inequality and $\E|s_{1,1} |^8 < \infty$) gives $T^{2} \sum_{t=0}^{T-1} \E [| (\E_{t+1} - \E_t) [S_T^{\sf H} D_T S_T]_{i,j} |^4 ] \to 0$, $1\leq i,j \leq K$, which proves Lyapunov's condition. Denoting $D_T=[d_{ij}]$,
\begin{align*} 
T Z_t = d_{t+1,t+1} \tr C({\bs s}_{t+1} {\bs s}_{t+1}^{\sf H} - I_K ) + 2 \Re\Bigl(\sum_{i,j=1}^K c_{i,j} \sum_{k=1}^t s_{k,j}^* \, s_{t+1,i} \, d_{k, t+1}\Bigr) .
\end{align*}
Using the independence of the $s_{i,j}$ and the moments $\E s_{1,1} = 0$, $\E| s_{1,1}|^2 = 1$, and $\E [s_{1,1}^u (s_{1,1}^*)^v] = 0$ for $u \neq v$, we obtain
\begin{align*} 
T^2\E_t Z_t^2 = d_{t+1,t+1}^2 \Bigl( \tr C^2 + \kappa \sum_{k=1}^K c_{kk}^2 \Bigr) + 2 \sum_{i,j,n=1}^K c_{i,j} c_{n,i} \sum_{k,\ell=1}^t s_{k,j}^* \, s_{\ell,n} \, d_{k,t+1} d_{t+1,\ell}. 
\end{align*} 
Letting $\check D_T = [ d_{ij} \1_{i>j}]$, we have
\begin{align*} 
T \sum_{t=0}^{T-1} \E_t Z_t^2 = \Bigl( \tr C^2 + \kappa \sum_{k=1}^K c_{kk}^2 \Bigr) \frac{1}{T} \tr (\diag(D_T))^2  
+ \frac{2}{T} \tr C S_T^{\sf H} \check D_T^{\sf H} \check D_T S_T C. 
\end{align*} 
Using \cite[Lemma 2.7]{BaiSil'98} and \cite[Lemma 3]{KammKhaHac'09} (or \cite[P.~278]{Niko'02}), we then get 
\begin{equation*}
\frac{1}{T} \tr C S_T^{\sf H} \check D_T^{\sf H} \check D_T S_T C 
- \tr C^2 \frac 1T \tr \check D_T^{\sf H} \check D_T 
\toprobalong 
0.
\end{equation*}
We finally get the result by observing that 
\begin{equation*}
\frac 2T \tr \check D_T^{\sf H} \check D_T = 
\frac 1T \tr D_T^2 - \frac 1T \tr (\diag(D_T))^2 .
\end{equation*}

\subsection{Proof of Lemma \ref{lm-diag}} 
\label{prf-diag} 
\cite[Lemma 3.1]{ChapCouHac'12} shows that for any compact $K \subset \R - \support(\bmu)$, there exists $C > 0$ such that 
\begin{equation*}
\forall T \ \text{large}, \ \forall t \in \support(\nu_T), \
\inf_{x \in K} | 1 + c_T m_T(x) t | > C 
\end{equation*} 
and hence $\liminf_T \inf_{t \in \support(\nu_T)} | 1 + c_T m_T(\rho_{k,T}) t | > 0$.  
It results that $\limsup_T \| D_T \| < \infty$. Furthermore, since 
\begin{equation*}
\frac 1T \tr (D_T^2) = 
\int \frac{p_k^2 m_T(\rho_{k,T})^2}
{(1 + c_T m_T(\rho_{k,T}) t)^2} \nu_T(dt) 
\end{equation*}
the first convergence in the statement of Lemma~\ref{lm-diag} holds true. \\ 
As for the second convergence, recall that $R_T = [ r_{t-n} ]_{1\leq t,n\leq T}$, with $\sum_t|r_t|<\infty$, and define the Toeplitz matrix $\Gamma_T \triangleq [ \gamma_{t-n} ]_{1\leq t,n\leq T}$ where $\gamma_\ell = {\bdelta}_\ell + \bc \bm(\brho_k) r_\ell$. Observe that
$D_T = p_k m_T(\rho_{k,T}) \Gamma_T^{-1}$. 
Let $[\cdot]_T$ be the modulo-$T$ operator, and let $\widetilde\Gamma_T = [ \gamma_{[t-n]_T} ]_{1\leq t,n\leq T}$ be a circulant matrix associated with $\Gamma_T$. By \cite[Lemma~3.1]{ChapCouHac'12} again, $\liminf_T \inf_{u\in [0,1]} (1 + c_T m_T(\rho_{k,T}) | \bp(\exp(2\imath\pi u)) |^2 ) > 0$, hence $\sup_T \| \widetilde\Gamma_T \| < \infty$. It results that $T^{-1} \| \Gamma_T^{-1} - \widetilde\Gamma_T^{-1} \|_{\text{fro}}^2 \to 0$, with $\| \cdot \|_{\text{fro}}$ the Frobenius norm \cite[Th.~5.2]{Gray'06}. On the other hand, since $\widetilde\Gamma_T$ is circulant, its eigenvector matrix is the Fourier $T \times T$ matrix, so that we can show 
\begin{equation*}
\diag(\widetilde\Gamma_T^{-1}) = 
\Bigl( \frac 1T \sum_{t=0}^{T-1} 
\frac{1}{1 + c_T m_T(\rho_{k,T}) | \bp(\exp(2\imath\pi t/T)) |^2} \Bigr) I_T . 
\end{equation*} 
The lemma is obtained by combining these last two results. 

\subsection{Proof of Lemma \ref{lm:fq-totale}} 
\label{prf-totale} 
We essentially show that we can replace the $B_{k,T}$ by $\sqrt{p_k} S_{k,T}$ with $S_T = [ S_{1,T}, \ldots, S_{t,T} ]$, similar to $B_T$. Since $\theta_i \neq \theta_j$ if $i\neq j$, from the definition of the vector function ${\bf a}(\theta)$, we have $[H_T^{\sf H} H_T ]_{k,\ell} - \bdelta_{k\ell} = \ba_T(\theta_k)^{\sf H} \ba_T(\theta_\ell) - \bdelta_{k\ell} = {\mathcal O}(1/T)$. Hence, $(H_T^{\sf H} H_T)^{1/2} \triangleq I_K + E_T$ where $\| E_T \| = {\mathcal O}(1/T)$. Given any sequence $D_T$ of deterministic matrices such that $\sup_T \| D_T \| < \infty$, it can be seen by a moment derivation with respect to the law of $S_T$ that $\E | [ B_T^{\sf H} D_T B_T - P^{1/2} S_T^{\sf H} D_T S_T P^{1/2} ]_{k,\ell} | = {\mathcal O}(1/T)$ for any $k,\ell \leq K$. Hence, by Markov's inequality, $\sqrt{T} ( B_T^{\sf H} D_T B_T - P^{1/2} S_T^{\sf H} D_T S_T P^{1/2})\toprobashort 0$. Replacing $D_T$ with any of the matrices $p_k m_T(\rho_{k,T}) ( I_T + c_T m_T(\rho_{k,T}) R_T)^{-1}$, we get from Lemma~\ref{lm-diag} that $\sup_T \| D_T \| < \infty$. Therefore, the $B_{k,T}$ can be replaced with the $\sqrt{p_k} S_{k,T}$. The result is then obtained upon applying Lemmas~\ref{tlc:fq} and \ref{lm-diag} and recalling that, for $k=1,\ldots, t$, the $S_{k,T}$ are independent.

\bibliographystyle{IEEEtran}
\bibliography{Bibliography} 

\begin{thebibliography}{10}
\providecommand{\url}[1]{#1}
\csname url@samestyle\endcsname
\providecommand{\newblock}{\relax}
\providecommand{\bibinfo}[2]{#2}
\providecommand{\BIBentrySTDinterwordspacing}{\spaceskip=0pt\relax}
\providecommand{\BIBentryALTinterwordstretchfactor}{4}
\providecommand{\BIBentryALTinterwordspacing}{\spaceskip=\fontdimen2\font plus
\BIBentryALTinterwordstretchfactor\fontdimen3\font minus
  \fontdimen4\font\relax}
\providecommand{\BIBforeignlanguage}[2]{{%
\expandafter\ifx\csname l@#1\endcsname\relax
\typeout{** WARNING: IEEEtran.bst: No hyphenation pattern has been}%
\typeout{** loaded for the language `#1'. Using the pattern for}%
\typeout{** the default language instead.}%
\else
\language=\csname l@#1\endcsname
\fi
#2}}
\providecommand{\BIBdecl}{\relax}
\BIBdecl

\bibitem{URK67}
H.~Urkowitz, ``{Energy detection of unknown deterministic signals},''
  \emph{Proceedings of the IEEE}, vol.~55, no.~4, pp. 523--531, 1967.

\bibitem{Akai'74}
H.~Akaike, ``A new look at the statistical model identification,'' \emph{IEEE
  Transactions on Automatic Control}, vol. AC-19, no.~6, pp. 716--723, 1974.

\bibitem{Riss'78}
J.~Rissanen, ``Modeling by shortest data description,'' \emph{Automatica},
  vol.~14, pp. 465--471, 1978.

\bibitem{WaxKai'85}
M.~Wax and T.~Kailath, ``Detection of signals by information theoretic
  criteria,'' \emph{IEEE Trans.~Acoustics, Speech, Signal Processing}, vol.~33,
  no.~2, pp. 387--392, 1985.

\bibitem{Schm'86}
R.~O. Schmidt, ``Multiple emitter location and signal parameter estimation,''
  \emph{IEEE Trans.~Antenna and Propag.}, vol.~34, no.~3, pp. 276--280, 1986.

\bibitem{NadaEde'08}
\BIBentryALTinterwordspacing
R.~R. Nadakuditi and A.~Edelman, ``Sample eigenvalue based detection of
  high-dimensional signals in white noise using relatively few samples,''
  \emph{IEEE Transactions on Signal Processing}, vol.~56, no.~7, pp.
  2625--2638, 2008. [Online]. Available:
  \url{http://dx.doi.org/10.1109/TSP.2008.917356}
\BIBentrySTDinterwordspacing

\bibitem{KritNad'09}
S.~Kritchman and B.~Nadler, ``Non-parametric detection of the number of
  signals: hypothesis testing and random matrix theory,'' \emph{IEEE
  Transactions on Signal Processing}, vol.~57, no.~10, pp. 3930--3941, 2009.

\bibitem{KritNad'08}
------, ``Determining the number of components in a factor model from limited
  noisy data,'' \emph{Chemometrics and Intelligent Laboratory Systems},
  vol.~94, no.~1, pp. 19--32, 2008.

\bibitem{COU10b}
R.~Couillet, J.~W. Silverstein, Z.~D. Bai, and M.~Debbah, ``{Eigen-inference
  for energy estimation of multiple sources},'' \emph{IEEE Transactions on
  Information Theory}, vol.~57, no.~4, pp. 2420--2439, 2011.

\bibitem{mestre2008improved}
X.~Mestre, ``Improved estimation of eigenvalues and eigenvectors of covariance
  matrices using their sample estimates,'' \emph{IEEE Transactions on
  Information Theory}, vol.~54, no.~11, pp. 5113--5129, 2008.

\bibitem{BessKraSch'06}
O.~Besson, S.~Kraut, and L.~L. Scharf, ``Detection of an unknown rank-one
  component in white noise,'' \emph{IEEE Transactions on Signal Processing},
  vol.~54, no.~7, pp. 2835--2839, 2006.

\bibitem{card-08}
L.~S. Cardoso, M.~Debbah, P.~Bianchi, and J.~Najim, ``Cooperative spectrum
  sensing using random matrix theory,'' in \emph{Wireless Pervasive Computing,
  2008. ISWPC 2008. 3rd International Symposium on}.\hskip 1em plus 0.5em minus
  0.4em\relax IEEE, 2008, pp. 334--338.

\bibitem{penna-cl09}
F.~Penna, R.~Garello, and M.~Spirito, ``Cooperative spectrum sensing based on
  the limiting eigenvalue ratio distribution in {W}ishart matrices,''
  \emph{Communications Letters, IEEE}, vol.~13, no.~7, pp. 507--509, 2009.

\bibitem{BianDebMai'11}
P.~Bianchi, M.~Debbah, M.~Maida, and J.~Najim, ``Performance of statistical
  tests for single-source detection using random matrix theory,'' \emph{IEEE
  Transactions on Information Theory}, vol.~57, no.~4, pp. 2400--2419, 2011.

\bibitem{nad-icc11}
B.~Nadler, F.~Penna, and R.~Garello, ``Performance of eigenvalue-based signal
  detectors with known and unknown noise level,'' in \emph{Communications
  (ICC), 2011 IEEE International Conference on}.\hskip 1em plus 0.5em minus
  0.4em\relax IEEE, 2011, pp. 1--5.

\bibitem{NadaSil'10}
R.~R. Nadakuditi and J.~W. Silverstein, ``Fundamental limit of sample
  generalized eigenvalue based detection of signals in noise using relatively
  few signal-bearing and noise-only samples,'' \emph{IEEE Journal of Selected
  Topics in Signal Processing}, vol.~4, no.~3, pp. 468--480, 2010.

\bibitem{nad-joh-ssp11}
B.~Nadler and I.~M. Johnstone, ``Detection performance of {R}oy's largest root
  test when the noise covariance matrix is arbitrary,'' in \emph{Statistical
  Signal Processing Workshop (SSP), 2011 IEEE}.\hskip 1em plus 0.5em minus
  0.4em\relax IEEE, 2011, pp. 681--684.

\bibitem{MarcPas'67}
\BIBentryALTinterwordspacing
V.~A. Mar\v{c}enko and L.~A. Pastur, ``Distribution of eigenvalues for some
  sets of random matrices,'' \emph{Mathematics of the USSR-Sbornik}, vol.~1,
  no.~4, pp. 457--483, 1967. [Online]. Available:
  \url{http://dx.doi.org/10.1070/SM1967v001n04ABEH001994}
\BIBentrySTDinterwordspacing

\bibitem{SilvBai'95}
J.~W. Silverstein and Z.~D. Bai, ``On the empirical distribution of eigenvalues
  of a class of large-dimensional random matrices,'' \emph{J. Multivariate
  Anal.}, vol.~54, no.~2, pp. 175--192, 1995.

\bibitem{SilvCho'95}
\BIBentryALTinterwordspacing
J.~W. Silverstein and S.~Choi, ``Analysis of the limiting spectral distribution
  of large-dimensional random matrices,'' \emph{J. Multivariate Anal.},
  vol.~54, no.~2, pp. 295--309, 1995. [Online]. Available:
  \url{http://dx.doi.org/10.1006/jmva.1995.1058}
\BIBentrySTDinterwordspacing

\bibitem{BaiSil'98}
Z.~D. Bai and J.~W. Silverstein, ``No eigenvalues outside the support of the
  limiting spectral distribution of large-dimensional sample covariance
  matrices,'' \emph{Ann. Probab.}, vol.~26, no.~1, pp. 316--345, 1998.

\bibitem{HornJoh'90}
R.~A. Horn and C.~R. Johnson, \emph{Matrix analysis}.\hskip 1em plus 0.5em
  minus 0.4em\relax Cambridge: Cambridge University Press, 1990.

\bibitem{ChapCouHac'12}
F.~Chapon, R.~Couillet, W.~Hachem, and X.~Mestre, ``On the isolated eigenvalues
  of large {G}ram random matrices with a fixed rank deformation,'' submitted.
  [Online] arXiv:1207.0471.

\bibitem{HachLouMes'11}
W.~Hachem, P.~Loubaton, X.~Mestre, J.~Najim, and P.~Vallet, ``A subspace
  estimator for fixed rank perturbations of large random matrices,''
  \emph{accepted for publication in the Journal of Multivariate Analysis},
  2012, [online] arXiv/1106.1497.

\bibitem{pastur2011eigenvalue}
L.~Pastur and M.~{\^S}erbina, \emph{Eigenvalue distribution of large random
  matrices}.\hskip 1em plus 0.5em minus 0.4em\relax American Mathematical Soc.,
  2011, vol. 171.

\bibitem{Gray'06}
R.~M. Gray, \emph{Toeplitz and circulant matrices: A review}.\hskip 1em plus
  0.5em minus 0.4em\relax Now Pub., 2006.

\bibitem{GrenSze'84}
U.~Grenander and G.~Szeg{\H{o}}, \emph{Toeplitz forms and their applications},
  2nd~ed.\hskip 1em plus 0.5em minus 0.4em\relax New York: Chelsea Publishing
  Co., 1984.

\bibitem{elkar'07}
\BIBentryALTinterwordspacing
N.~El~Karoui, ``Tracy-{W}idom limit for the largest eigenvalue of a large class
  of complex sample covariance matrices,'' \emph{Ann. Probab.}, vol.~35, no.~2,
  pp. 663--714, 2007. [Online]. Available:
  \url{http://dx.doi.org/10.1214/009117906000000917}
\BIBentrySTDinterwordspacing

\bibitem{VAN00}
A.~W. {Van~der~Vaart}, \emph{{Asymptotic Statistics}}.\hskip 1em plus 0.5em
  minus 0.4em\relax New York: Cambridge University Press, 2000.

\bibitem{bickel2008regularized}
P.~J. Bickel and E.~Levina, ``Regularized estimation of large covariance
  matrices,'' \emph{The Annals of Statistics}, vol.~36, no.~1, pp. 199--227,
  2008.

\bibitem{bai-sil-clt04}
\BIBentryALTinterwordspacing
Z.~D. Bai and J.~W. Silverstein, ``C{LT} for linear spectral statistics of
  large-dimensional sample covariance matrices,'' \emph{Ann. Probab.}, vol.~32,
  no.~1A, pp. 553--605, 2004. [Online]. Available:
  \url{http://dx.doi.org/10.1214/aop/1078415845}
\BIBentrySTDinterwordspacing

\bibitem{bro-dav-91}
P.~J. Brockwell and R.~A. Davis, \emph{Time series: theory and methods}, ser.
  Springer Series in Statistics.\hskip 1em plus 0.5em minus 0.4em\relax New
  York: Springer, 2006, reprint of the second (1991) edition.

\bibitem{BhanGirKok'07}
\BIBentryALTinterwordspacing
R.~J. Bhansali, L.~Giraitis, and P.~S. Kokoszka, ``Convergence of quadratic
  forms with nonvanishing diagonal,'' \emph{Statist. Probab. Lett.}, vol.~77,
  no.~7, pp. 726--734, 2007. [Online]. Available:
  \url{http://dx.doi.org/10.1016/j.spl.2006.11.007}
\BIBentrySTDinterwordspacing

\bibitem{KammKhaHac'09}
A.~Kammoun, M.~Kharouf, W.~Hachem, and J.~Najim, ``A {C}entral {L}imit
  {T}heorem for the {SINR} at the {LMMSE} estimator output for
  large-dimensional signals,'' \emph{IEEE Transactions on Information Theory},
  vol.~55, no.~11, pp. 5048--5063, 2009.

\bibitem{Bill'95}
P.~Billingsley, \emph{Probability and measure}, 3rd~ed., ser. Wiley Series in
  Probability and Mathematical Statistics.\hskip 1em plus 0.5em minus
  0.4em\relax New York: John Wiley \& Sons Inc., 1995.

\bibitem{Niko'02}
N.~K. Nikolski, \emph{Operators, functions, and systems: an easy reading.
  {V}ol. 2}, ser. Mathematical Surveys and Monographs.\hskip 1em plus 0.5em
  minus 0.4em\relax Providence, RI: American Mathematical Society, 2002,
  vol.~93.

\end{thebibliography}

\end{document}